\documentclass[preprint,10pt]{elsarticle}
\usepackage[utf8]{inputenc}
\usepackage{myStyle}
\usepackage{stfloats}
\usepackage{my_shapes}

\theoremstyle{proposition}
\newtheorem{proposition}{Proposition}
\theoremstyle{theorem}
\newtheorem{theorem}{Theorem}

\journal{Signal Processing}
\begin{document}

\begin{frontmatter}
\title{GLOSS: Tensor-Based Anomaly Detection in Spatiotemporal Urban Traffic Data\blfootnote{This work was supported in part by NSF CCF-1615489 and DMS-1924724.}}
\author[1]{Seyyid Emre Sofuoglu}
\author[1]{Selin Aviyente}
\address[1]{{Department of Electrical and Computer Engineering  Michigan State University},
            {East Lansing},
            {MI},
            {48824}. 
            }

\begin{abstract}
Anomaly detection in spatiotemporal data is a challenging problem encountered in a variety of applications including hyperspectral imaging, video surveillance and urban traffic monitoring. In the case of urban traffic data, anomalies refer to unusual events such as traffic congestion and unexpected crowd gatherings. Detecting these anomalies is challenging due to the dependence of anomaly definition on time and space. In this paper, we introduce an unsupervised tensor-based anomaly detection method for spatiotemporal urban traffic data. The proposed method assumes that the anomalies are sparse and temporally continuous, {i.e.},  anomalies appear as spatially contiguous groups of locations that show anomalous values consistently for a short duration of time. Furthermore, a manifold embedding approach is adopted to preserve the local geometric structure of the data across each mode. The proposed framework, Graph Regularized Low-rank plus Temporally Smooth Sparse decomposition (GLOSS), is formulated as an optimization problem and solved using alternating method of multipliers (ADMM). The resulting algorithm is shown to converge and be robust against missing data and noise. The proposed framework is evaluated on both synthetic and real spatiotemporal urban traffic data and compared with baseline methods.
\end{abstract}


\end{frontmatter}

\section{Introduction}
Large volumes of spatiotemporal (ST) data are increasingly collected and studied in diverse domains, including climate science, social sciences, neuroscience, epidemiology \cite{bhunia2013spatial}, transportation \cite{djenouri2019survey}, mobile health, and Earth sciences \cite{zhang2016tensor}. One emerging application of interest in spatiotemporal data is anomaly detection. Detecting anomalies can help us identify interesting but rare phenomena, e.g. abnormal flow of crowds, traffic congestion in traffic monitoring or hot-spots for monitoring outbreaks in infectious diseases.


The definition of anomaly and the suitability of a particular method is determined by the application. In this paper, we focus on urban anomaly detection \cite{zhang2020urban, zhang2019decomposition, li2019tensor, lin2018anomaly, zhang2018detecting}, where anomalies correspond to incidental events that occur rarely, such as irregularity in traffic volume, unexpected crowds, etc. Urban data are spatiotemporal data collected by mobile devices or distributed sensors in cities and are usually associated with timestamps and location tags. Detecting and predicting urban anomalies are of great importance to policymakers and governments for understanding city-scale human mobility and activity patterns, inferring land usage and region functions and discovering traffic problems \cite{zhang2020urban, celes2019crowd}. The main challenges in the detection of urban anomalies based on spatio-temporal data include the scarcity of anomalies and the dependence of what constitutes an anomaly on spatial and temporal factors. 

In this paper, we will focus on urban event data such as trip records which comprise  discrete events occurring at point locations and times. Different data structures have been used to represent urban data including time series, matrices, graphs and tensors \cite{zhang2020urban}. However, it is important for the selected data structure to capture the spatial and temporal relations within the data. For this reason, in this paper we will represent spatiotemporal urban event data using higher-order tensors  to capture both temporal and spatial dependencies. In particular, we propose a robust tensor decomposition method for unsupervised anomaly detection in urban event data. 
The main contributions of the proposed method are:
\begin{itemize}
     \item A low-rank plus sparse tensor decomposition similar to higher order robust PCA (HoRPCA) \cite{goldfarb2014robust} is adopted to decompose the urban traffic data into normal and anomalous components. The normal component corresponds to the low-rank structure of the observed tensor while the sparse tensor captures the anomalies. Unlike regular HoRPCA, we employ a weighted nuclear norm definition to emphasize the difference in low-rank structure across modes. 
    \item The proposed low-rank plus sparse tensor decomposition is modified to account for the characteristics of urban traffic data. As anomalies tend to last for periods of time, we impose temporal smoothness on the sparse part of the tensor through total variation regularization. This regularization ensures that instantaneous changes in the data, which may be due to errors in sensing, are not mistaken for actual anomalies. This formulation leads to our first algorithm; low-rank plus temporally smooth sparse (LOSS) tensor decomposition. 
    \item We introduce a graph regularized version of LOSS to exploit the geometric structure of the data. This new algorithm, named GLOSS, preserves the local geometry across each mode while still obtaining a low-rank approximation to the observed spatiotemporal tensor.
    \item Our optimization framework is formulated in a flexible manner such that it solves a tensor completion problem simultaneously with the anomaly detection problem. As such, the resulting algorithm is robust against missing data.
    \item Finally, as the proposed framework extracts low-dimensional spatiotemporal projections of the original data to separate anomalies from normal events, the resulting spatiotemporal features can be used as input to any conventional statistical anomaly detection or scoring algorithm.
\end{itemize}
 
The rest of the paper is organized as follows. In Section, \ref{sec:related}, we review some of the related work in spatiotemporal anomaly detection, in particular tensor based anomaly detection methods. In Section \ref{sec:back}, we provide background on tensor operations and tensor norms. In Section \ref{sec:meth}, we formulate the optimization problem and propose an ADMM based solution. In Section \ref{sec:Exp}, we describe the experimental settings both with synthetic and real data and compare the proposed method with baseline anomaly detection methods as well as tensor based methods.

\section{Related Work}
\label{sec:related}

Spatiotemporal data bring with themselves unique challenges to anomaly detection problem due to the autocorrelation structure of the normal points. Different  algorithms have been proposed to detect three types of ST anomalies: point anomalies; trajectory anomalies and group anomalies.  Point anomalies are defined as spatiotemporal outliers that break the natural ST autocorrelation structure of the normal points.  Most ST point anomaly detection algorithms such  as ST-DBSCAN  \cite{kut2006spatio} assume homogeneity in neighborhood properties across space and time, which can be violated in the presence of ST heterogeneity. Trajectory anomalies are usually detected by computing pairwise similarities among trajectories and identifying trajectories that are spatially distant from the others \cite{lee2007trajectory}. Finally, group anomalies appear in ST data as spatially contiguous groups of locations that show anomalous values consistently for a short duration of time stamps. The urban anomalies considered in this paper fall into this last category. Most approaches for detecting group anomalies in ST urban traffic data decompose the anomaly detection problem by first treating the spatial and temporal properties of the outliers independently, and  then merging together in a post-processing step \cite{faghmous2013parameter,wu2008spatio}. Traditional vector based methods focus on building appropriate time series model to characterize
the inherent spatiotemporal patterns, dependencies, and the
generative mechanism of the data such as  time-varying autoregressive (TVAR) \cite{bringmann2017changing,bringmann2018modeling} and
switching Kalman filtering/smoothing (SKF/SKS) \cite{murphy1998switching,nguyen2018anomaly},
which are built on traditional dynamic linear models
(DLMs) \cite{li2015trend}. However,  as the number of sensors increases, scalability becomes a critical issue and these methods become unreliable. For these reasons, one natural approach to address ST group  anomaly detection has been to use tensor decomposition.

Low-rank tensor decomposition and completion have been proposed as suitable approaches to anomaly detection in spatiotemporal data as these methods are a natural extension of spectral anomaly detection techniques from vector to multi-way data \cite{fanaee2016tensor,li2015low,geng2014high,zhang2016tensor,chen2016fine,lin2018anomaly,li2019tensor,xu2019anomaly}.  Most of the existing tensor based methods are supervised or semi-supervised and focus on dimensionality reduction and feature extraction. In this line of work, tensor decomposition is first applied to normal tensor samples and the factor matrices are fed to a classifier to build a model for normal activity. This model is then used to predict the label of observation in the test factor matrix.  These methods require labeled historical data which is not suitable for online real-time anomaly detection. Unsupervised tensor based anomaly detection methods \cite{xu2019anomaly}, on the other hand, aim to learn spatiotemporal features within a representation learning framework \cite{fanaee2016event,xu2019anomaly,shi2015stensr}. The learned features, {i.e.} factor matrices or core tensors, are then used to detect anomalies by  monitoring the reconstruction error at each time point \cite{papalexakis2012network,papalexakis2014spotting,sun2006beyond,xu2019anomaly} or by applying well-known statistical tests to the extracted multivariate features \cite{fanaee2016event,zhang2016tensor}. 

The existing tensor based anomaly detection methods have multiple shortcomings. First, they rely on well-known tensor decomposition models such as Tucker \cite{fanaee2016event,zhang2016tensor,xu2019anomaly} and CP \cite{li2019tensor}. These methods use tensor decomposition to obtain a low-dimensional projection of the data without taking the particular structure of anomalies into account. The proposed method incorporates the sparseness of anomalies into the model by assuming that the anomalies lie in the sparse part of the tensor rather than in the low-rank part. Moreover, by incorporating tensor completion, the proposed method is robust against missing input data. Second, prior work on HoRPCA within the framework of anomaly detection  \cite{li2015low,geng2014high} does not ensure temporal smoothness for the detected anomalies. However, in urban data, anomalies typically last for some time and are not instantaneous. Finally, the existing tensor based anomaly detection methods are limited to anomalies that lie in low dimensional linear subspaces. The proposed method utilizes a graph Laplacian regularization approach to  preserve the intrinsic manifold structure of high-dimensional data and capture the nonlinearities in the data structure. Although graph regularized tensor decomposition has been used in both unsupervised and supervised learning applications \cite{li2016mr, jiang2018image, qiu2020generalized}, the graph regularization proposed in this paper is novel in some key ways. In prior work, most of the focus has been on supervised learning and the graph regularization is with respect to the mode-$N$ unfolding matrix, where mode-$N$ corresponds to the sample mode. Unlike prior work, we consider regularization across all modes to account for the heterogeneity in the data. Moreover, there have been different interpretations of what constitutes the low-dimensional representation of the tensor objects. While some papers consider the decomposition factors along object dimension as  the low-dimensional projections \cite{qiu2020generalized}, others have considered the core tensors as the low-dimensional features \cite{li2016mr}. In this paper, we consider the low-rank approximation of the observed data tensor as the low-dimensional projection similar to \cite{sun2018graph}. 
\section{Background}
\label{sec:back}
Let $\Y\in\R^{I_1\times I_2\times \dots \times I_N}$ be a $N$-mode tensor, where $y_{i_1,i_2,\dots,i_N}$ denotes the $({i_1,i_2,\dots,i_N})^{th}$ element of the tensor $\Y$. Vectors obtained by fixing all indices of the tensor except the one that corresponds to $n$th mode are called mode-$n$ fibers and denoted as $\y_{i_1,\dots i_{n-1},i_{n+1},\dots i_N}\in\R^{I_n}$.

\noindent \textbf{Definition 1.} (Mode-$n$ product) The mode-$n$ product of a tensor $\mathcal{A}\in \mathbb{R}^{I_1\times ... I_n\times ...\times I_N}$ and a matrix ${\bf{U}} \in \mathbb{R}^{J\times I_n} $ is denoted as $\mathcal{B}=\mathcal{A} \times_{n} {\bf{U}}$ and defined as $b_{i_{1},\ldots,i_{n-1},j,i_{n+1}\ldots,i_{N}}=\sum_{i_{n}=1}^{I_{n}}a_{i_{1}\ldots,i_{n},\ldots,i_{N}}u_{j,i_{n}}$, where $\mathcal{B} \in \mathbb{R}^{I_1\times ...\times I_{n-1} \times J \times I_{n+1}\times ...\times I_N}$. 

\noindent \textbf{Definition 2.} (Mode-$n$ unfolding) The mode-$n$ unfolding of a tensor $\Y$ is defined as $\Y_{(n)}\in \R^{I_n\times \prod_{n'=1, n'\neq n}^NI_{n'}}$ where the mode-n fibers of the tensor $\Y$ are the columns of $\Y_{(n)}$ and the remaining modes are organized accordingly along the rows.

\noindent \textbf{Definition 3.} (Mode-$n$ Concatenation) Mode-$n$ concatenation unfolds the input tensors along mode-$n$ and stacks the unfolded matrices across rows:
\begin{gather}
    \text{cat}_n(\Y,\Y) = 
    \begin{bmatrix}
        \Y_{(n)}^\top&\Y_{(n)}^\top
    \end{bmatrix}^\top.
\end{gather}
If the input are a set of tensors, {e.g.} $\{\TT\} = \{\TT_1, \TT_2, \dots, \TT_M\}$, where $\TT_m \in \R^{I_1\times \dots\times I_N}, \forall m\in \{1,\dots,M\}$, then $\text{cat}_n(\{\TT\})$ stacks all mode $n$ unfoldings of tensors $\{\TT\}$ across rows into a matrix of size $MI_n \times \prod_{n'=1, n'\neq n}^NI_{n'}$.

\noindent \textbf{Definition 4.} (Tensor norms)  In this paper, we employ three different tensor norms. Frobenius norm of a tensor is defined as $\|\Y\|_F=\sqrt{\sum_{i_1,i_2,\dots,i_N} y_{i_1,i_2,\dots,i_N}^2}$.  $\ell_1$ norm of a tensor is defined as $\|\Y\|_1=\sum_{i_1,i_2,\dots,i_N}|y_{i_1,i_2,\dots,i_N}|$. Finally, in this paper the nuclear norm of a tensor is defined as the weighted sum of the nuclear norms of all mode-n unfoldings of a tensor, namely $\|\Y\|_*=\sum_{n=1}^N\psi_n \|\Y_{(n)}\|_*$, where $\psi_{n}$s are the weights corresponding to each mode  \cite{tomioka2010extension,goldfarb2014robust}. 

\noindent \textbf{Definition 5.} (Mode-n Graph Adjacency and Laplacian Matrices) Given the tensor $\Y$, we would like to construct graphs that model the local geometry or neighborhood information across each mode. The idea is to connect each row of the mode-$n$ unfolding with its $k$-nearest neighborhoods. Let $\G^n=(\V^n, \mathcal{E}^n)$ be a graph where $\mathcal{E}^n$ correspond to edges and $\V^n$ are the vertices (rows of mode-$n$ unfolding of $\Y$). The mode-n adjacency matrix is defined by quantifying the similarity, $w_{s,s'}^n$, between rows $s$ and $s'$ of mode-$n$ unfolding of $\Y$:
\vspace{-.5em}
\begin{gather}
    w_{ss'}^n=
    \begin{cases}
        e^{-\frac{\|\Y_{(n),s}-\Y_{(n),s'}\|_F^2}{2\sigma}}, & \text{if }  \Y_{(n),s}\in \Nb_k(\Y_{(n),s'}) \\
        &\text{ or } \Y_{(n),s'}\in \Nb_k(\Y_{(n),s})\\
        0, & \text{otherwise}
    \end{cases},
\end{gather}
where $\Nb_k(\Y_{(n),s})$ is the Euclidean k-nearest neighborhood of the $s^{th}$ row of $\Y_{(n)}$, $\Y_{(n),s}$. The mode-$n$ graph Laplacian $\Phi^n$ is then defined as $\Phi^n=D^n-W^n$, where $D^n$ is the diagonal degree matrix with entries $d_{i,i}^n = \sum_{i'=1}^{I_n} w_{i,i'}^n$.

\noindent \textbf{Definition 6.} (Support Set) Let $\Omega$ be a support set defined for tensor $\Y$, {i.e.} $\Omega \in \{1,\dots,I_1\}\times \{1,\dots,I_2\}\times\dots\times\{1,\dots,I_N\}$. The projection operator on this support set, $\P_\Omega$, is defined  as:
\begin{gather}
    \P_\Omega[\Y]_{i_1,i_2,\dots,i_N} = 
    \begin{cases}
        \Y_{i_1,i_2,\dots,i_N}, & (i_1,i_2,\dots,i_N)\in\Omega,\\
        0, & \text{otherwise.}
    \end{cases}
\end{gather}
The orthogonal complement of the operator $\P_\Omega$ is defined in a similar manner as:
\begin{gather}
    \P_{\Omega^\perp}[\Y]_{i_1,i_2,\dots,i_N} = 
    \begin{cases}
        \Y_{i_1,i_2,\dots,i_N}, & (i_1,i_2,\dots,i_N)\notin\Omega,\\
        0, & \text{otherwise.}
    \end{cases}
\end{gather}

\section{Methods}
\label{sec:meth}
In the following discussions, for the sake of simplicity, we model spatiotemporal data as a four mode tensor $\Y\in\R^{I_1\times I_2\times I_3\times I_4}$. The first mode is denoted as the temporal mode and corresponds to hours in a day. The second mode corresponds to the days of a week as urban traffic activity shows highly similar patterns on the same days of different weeks. The third mode corresponds to the different weeks and the last mode corresponds to the spatial locations, such as stations for metro data, sensors for traffic data or zones for other urban data.

\vspace{-1em}
\subsection{Problem Statement}
Assuming that anomalies are rare events, our goal is to decompose $\Y$ into a low-rank part, $\L$, that corresponds to normal activity and a sparse part, $\S$, that corresponds to the anomalies. This model relies on the assumption that normal activity can be embedded into a lower dimensional subspace while anomalies are outliers. We also take into account the existence of missing elements from the data, {i.e.}, the observed data is $\P_{\Omega}[\Y]$. This goal can be formulated through:
\begin{gather}
    \min_{\L,\S}\|\L\|_*+\lambda\|\S\|_1, \quad \text{s.t. } \P_{\Omega}[\L+\S]=\P_{\Omega}[\Y],
\end{gather}
where $\lambda$ is the regularization parameter for sparsity.

Since urban anomalies tend to be temporally continuous, {i.e.}, smooth in the first mode, this assumption can be incorporated into the above formulation as: 
\begin{gather}
    \min_{\L,\S}\|\L\|_*+\lambda\|\S\|_1+\gamma\|\S\times_1 \Delta\|_1, \; \text{s.t. }\P_{\Omega}[\L+\S]=\P_{\Omega}[\Y],
    \label{eq:opt_lrt}
\end{gather}
where $\gamma$ is the regularization parameter for temporal smoothness and $\|\S\times_1 \Delta\|_1$ quantifies the sparsity of the projection of the tensor $\S$ onto the discrete-time differentiation operator along the first mode where $\Delta$ is defined as:
\begin{gather}
    \Delta = 
    \begin{bmatrix}
        1&-1&0&\dots&0\\
        0&1&-1&\dots&0\\
        \dots&\dots&\dots&\dots&\dots\\
        0&0&\dots&1&-1\\
        -1&0&\dots&0&1
    \end{bmatrix}.
\end{gather}

It is common to incorporate the relationships among data points as auxiliary information in addition to the low-rank assumption to improve the quality of tensor decomposition \cite{narita2012tensor}. This approach, also known as  manifold learning, is an effective dimensionality reduction technique leveraging geometric information. The intuitive idea behind manifold learning is that if two objects are close in the
intrinsic geometry of data manifold, they should be close
to each other after dimensionality reduction. For tensors, this usually reduces to forcing two similar objects  to behave similarly in the projected low-dimensional space through a graph
Laplacian term. In this paper, since we are trying to learn anomalies from a single tensor, we do not have tensor samples and their projections. Instead, we preserve the relationships between each mode unfolding of the tensor data as each mode corresponds to a different attribute of the data. 

\begin{figure}[H]
    \centering
    \begin{subfigure}[b]{\textwidth}
        \resizebox{\textwidth}{!}{
        \begin{tikzpicture}
            \draw (0, -6, 3)--(0, -6, 3) node[midway, above, sloped] {\Large $52$ (Weeks)};
            \draw[line width=0 pt] (2.5, -7, 0) -- (2.5, -7, 3) node[midway, below, sloped] {\Large $7$ (Days)};
            \draw (1, -7, 3) node[below] {\Large $24$ (Hours)};
            \draw (-.8, -7, 3)--(-.8,2,3) node[midway, above, sloped] {\Large $81$ (Zones)};
            \tensimg{1.25}{2.5}{-6}{-1}{0}{3}{54}{}{2.5}
            \tensimg{1.25}{2.5}{-2}{-1}{0}{3}{25}{}{2.5}
            \tensimg{1.25}{2.5}{1}{2}{0}{3}{1}{}{2.5}
            \draw (1,-4, 1.5) node [rotate = 90] {\Large ...};
            \draw (1, -8, 3) node[below] {\Huge \bf $\Y\in\R^{24\times 7\times52\times 81}$};
            \draw (.5*\xshift, -4) node {\Huge $\xrightarrow[\text{unfolding}]{\text{Mode}-n }$ };
            \draw (.8*\xshift+3.5, 1, 0) node {\includegraphics[width=7cm]{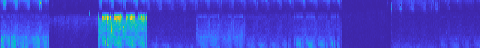}};
            \draw (.8*\xshift, 1) node[left] {\Large $\Y^{(1)}$};
            \draw (.8*\xshift+5, -1.75, 0) node {\includegraphics[width=10cm]{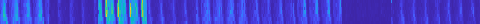}};
            \draw (.8*\xshift, -1.75) node[left] {\Large $\Y^{(2)}$};
            \draw (.8*\xshift+2, -4.75, 0) node {\includegraphics[width=4cm]{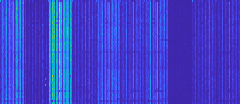}};
            \draw (.8*\xshift, -4.75) node[left] {\Large $\Y^{(3)}$};
            \draw (.8*\xshift+1.5, -8.5, 0) node {\includegraphics[width=3cm]{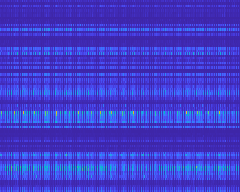}};
            \draw (.8*\xshift, -8.5) node[left] {\Large $\Y^{(4)}$};
            \draw (1.8*\xshift,-4) node {\huge $\xrightarrow[\text{Similarity Graph}]{\text{Mode-}n}$};
        \end{tikzpicture}
        \tdplotsetmaincoords{70}{0}    
        \begin{tikzpicture}[tdplot_main_coords]
            \CircGraph{4.5}{.2}{24}{blue!70}{+60}
            \draw (-3,0,4) node {\huge $\G^1$};
            \CircGraph{2}{.7}{7}{green!70}{20}
            \draw (-3,0,1) node {\huge $\G^2$};
            \RandGraph{4}{.2}{30}{yellow!70}{-2}
            \draw (-4,0,-2) node {\huge $\G^3$};
            \RandGraph{4}{.3}{44}{red!70}{-6}
            \draw (-4.2,0,-6) node {\huge $\G^4$};
        \end{tikzpicture}
        \begin{tikzpicture}
            \draw (-2,-4) node {\huge $\xrightarrow[\text{Laplacians }\Phi^n]{\text{Mode}-n}$};
            \draw (1,-4) node[rotate=90, align=center] {\Large $\min_{\L, \S} \|\L\|_*+\theta \sum_{n=1}^4\mathrm{tr}\left(\L_{(n)}^\top \Phi^n\L_{(n)}\right)+$\\\Large $\lambda\|\S\|_1+\gamma\sum_{i_2,i_3,i_4}\|\s_{i_2,i_3,i_4}\|_{\TV}$ };
            \draw (4,-4) node {\huge $\xrightarrow[]{\text{Algorithm \ref{alg:gloss}}}$};
            
            \draw[line width=0 pt] (2.5+\xshift, -7, 3) -- (2.5+\xshift, -7, 0) node[midway, below, sloped] {\Large $7$ (Days)};
            \draw (\xshift, -6, 3)--(\xshift, -6, 3) node[midway, above, sloped] {\Large $52$ (Weeks)};
            \draw (\xshift+1, -7, 3) node[below] {\Large $24$ (Hours)};
            \draw (-.8+\xshift, -7, 3)--(-.8+\xshift,2,3) node[midway, above, sloped] {\Large $81$ (Zones)};
            \tensimg{\xshift+1.25}{2.5}{-6}{-1}{0}{3}{54}{_lr}{2.5}
            \tensimg{\xshift+1.25}{2.5}{-2}{-1}{0}{3}{25}{_lr}{2.5}
            \tensimg{\xshift+1.25}{2.5}{1}{2}{0}{3}{1}{_lr}{2.5}
            \draw (1+\xshift,-4, 1.5) node [rotate = 90] {\Large ...};
            \draw (\xshift+1, -8, 3) node[below] {\Huge \bf $\L\in\R^{24\times 7\times52\times 81}$};
            \draw (1.3*\xshift+1,-4, 1.5) node {\Huge +};
            
            \draw[line width=0 pt] (2.5+1.6*\xshift, -7, 3) -- (2.5+1.6*\xshift, -7, 0) node[midway, below, sloped] {\Large $7$ (Days)};
            \draw (1.6*\xshift, -6, 3)--(1.6*\xshift, -6, 3) node[midway,above,sloped] {\Large $52$ (Weeks)};
            \draw (1+1.6*\xshift, -7, 3) node[below] {\Large $24$ (Hours)};
            \draw (3.3+1.6*\xshift, -7, 0)--(3.3+1.6*\xshift,2,0) node[midway, below, sloped] {\Large $81$ (Zones)};
            \tensimg{1.6*\xshift+1.25}{2.5}{-6}{-1}{0}{3}{54}{_sp}{2.5}
            \tensimg{1.6*\xshift+1.25}{2.5}{-2}{-1}{0}{3}{25}{_sp}{2.5}
            \tensimg{1.6*\xshift+1.25}{2.5}{1}{2}{0}{3}{1}{_sp}{2.5}
            \draw (1.6*\xshift+1,-4, 1.5) node [rotate = 90] {\Large ...};
            \draw (2+1.6*\xshift, -8, 3) node[below] {\Huge \bf $\S\in\R^{24\times 7\times52\times 81}$};
            \draw (2.1*\xshift,-4) node {\Huge $\xrightarrow[]{\S}$};
            \draw (2.4*\xshift,-4) node[block, text width=3cm] {\huge Anomaly Scoring: EE \cite{rousseeuw1999fast}, LOF \cite{breunig2000lof}};
        \end{tikzpicture}
    }
    \end{subfigure}
    \caption{An illustration of the proposed method (Left to right): Four mode tensor corresponding to urban traffic data; Mode-$n$ unfolding across each mode; Similarity graph construction across each mode; Optimization Algorithm and Anomaly Scoring. }
    \label{fig:illust}
    \vspace{-1em}
\end{figure}
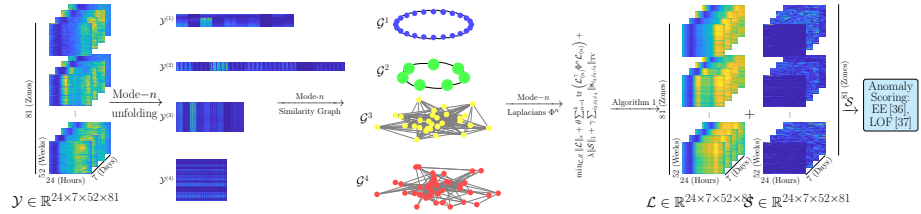
For this reason, in this paper, we consider  four graph Laplacian terms corresponding to the four modes (see Figure \ref{fig:illust}). For example, in urban traffic anomaly detection, mode-$1$ corresponds to the 24 hours within a day. When we build an adjacency matrix based on mode-$1$ unfoldings, i.e. $W^{1}$, we quantify the similarity of the traffic profiles between each one hour time frame regardless of day, week and spatial location.  Accordingly, the entries of the low-rank approximation unfolded across the first mode, $\mathcal{L}_{(1)}$, should preserve the original similarities, i.e. $\sum_{i,i'=1}^{I_{1}}w_{ii'}^{1} \sum_{p=1}^{I_{2}I_{3}I_{4}}(\mathcal{L}_{(1)}(i,p)-\mathcal{L}_{(1)}(i',p))^{2}$ should be minimized. Generalizing this idea across all modes and rewriting the distance metrics in terms of the graph Laplacians across each mode results in:
\begin{gather}
    \min_{\L,\S} \|\L\|_{*}+\frac{\theta}{2}\sum_{n=1}^4\sum_{i=1}^{I_n}\sum_{\substack{i'=1\\ i'\neq i}}^{I_n}\|\L_{(n),i}-\L_{(n),i'}\|_F^2 w_{ii'}^n+\lambda\|\S\|_1+\gamma\|\S\times_1 \Delta\|_1, \nonumber\\s.t. \quad\P_{\Omega}[\L+\S]=\P_{\Omega}[\Y],\nonumber
\end{gather}
where $\theta$ is the weight parameter for graph regularization.

The above optimization problem can equivalently be rewritten using trace norm to represent the graph regularization and total variation (TV) norm  across the temporal mode to describe temporal smoothness as:
\begin{gather}
    \min_{\L, \S} \|\L\|_*+\theta \sum_{n=1}^4\mathrm{tr}\left(\L_{(n)}^\top \Phi^n\L_{(n)}\right)+\lambda\|\S\|_1+\gamma\sum_{i_2,i_3,i_4}\|\s_{i_2,i_3,i_4}\|_{\TV},\nonumber\\
    \P_{\Omega}[\L+\S]=\P_{\Omega}[\Y],
    \label{eq:finopt}
\end{gather}
where $\|.\|_{\TV}$ denotes the total variation norm. The solution to the optimization problems \eqref{eq:opt_lrt} and  \eqref{eq:finopt} will be referred to as Low-rank plus Temporally Smooth Sparse (LOSS) Decomposition, and Graph Regularized Low-rank plus Temporally Smooth Sparse (GLOSS) Decomposition, respectively. Figure \ref{fig:illust} illustrates the different terms in the objective function in \eqref{eq:finopt}.

\subsection{Optimization}
The proposed objective function is convex. In prior work, ADMM has been effective at solving similar optimization problems \cite{goldfarb2014robust, aggarwal2016hyperspectral}. Thus, we follow a similar  approach for solving the optimization problem given in (\ref{eq:finopt}). To separate the minimization of TV, $\ell_1$ norm and graph regularization from each other, we introduce auxiliary variables $\Z, \W, \{\Lx\}:=\{\Lx^1,\Lx^2,\Lx^3,\Lx^4\}, \{\Laux\}:=\{\Laux^1,\Laux^2,\Laux^3,\Laux^4\}$ such that the optimization problem becomes:
\begin{gather}
    \min_{\L, \{\Lx\}, \{\Laux\}, \S, \Z, \W} \sum_{n=1}^4\left(\psi_n\|\Lx_{(n)}^n\|_*+\theta g(\Laux^n, \Phi^n)\right)+\lambda\|\S\|_1+\nonumber\\ \gamma\|\Z\|_{1}, \qquad \text{s.t. }\; \W=\S, \quad \Z=\W\times_1 \Delta, \nonumber\\
  \P_{\Omega}[\L+\S]=\P_{\Omega}[\Y], \Lx^n=\L, \Laux^n=\L, n \in \{1,2,3,4\},
    \label{eq:opt_f}
\end{gather}
where $g(\Laux^n,\Phi^n)=\mathrm{tr}\left({\Laux^n_{(n)}}^\top \Phi^n \Laux^n_{(n)}\right)$ is the graph regularization term for each auxiliary variable $\Laux^n$. To solve the above optimization problem, we propose using ADMM with partial augmented Lagrangian:
\begin{gather}
    \sum_{n=1}^4\left(\psi_n\|\Lx_{(n)}^n\|_*+\theta g(\Laux^n, \Phi^n)\right)+\lambda\|\S\|_1+\gamma\|\Z\|_{1}+\nonumber \\\frac{\beta_1}{2}\|\P_{\Omega}[\L+\S-\Y-\Lambda_1]\|_F^2+ \frac{\beta_2}{2}\sum_{n=1}^4\|\Lx^n-\L-\Lambda_2^n\|_F^2+\nonumber \\\frac{\beta_3}{2}\sum_{n=1}^4\|\L-\Laux^n-\Lambda_3^n\|_F^2 + \frac{\beta_4}{2}\|\W\times_1 \Delta-\Z-\Lambda_4\|_F^2+\nonumber \\ \frac{\beta_5}{2}\|\S-\W-\Lambda_5\|_F^2,
    \label{eq:aug_lagr}
\end{gather}
where $\Lambda_1, \Lambda_2^n, \Lambda_3^n, \Lambda_4, \Lambda_5 \in\R^{I_1\times I_2\times I_3\times I_4}$ are the Lagrange multipliers.

\textbf{1. $\L$ update:} The low-rank variable $\L$ can be updated using:
\begin{gather}
    \L^{t+1} = \argmin_{\L} \frac{\beta_1}{2}\|\P_{\Omega}[\L+\S^t-\Y-\Lambda_1^t]\|_F^2+ \nonumber \\\sum_{n=1}^4\left(\frac{\beta_2}{2}\|\Lx^{n,t}-\L-\Lambda_2^{n,t}\|_F^2+\frac{\beta_3}{2}\|\L-\Laux^{n,t}-\Lambda_3^{n,t}\|_F^2\right) ,
\end{gather}
which has the analytical solution:
\begin{gather}
    \P_{\Omega}[\L^{t+1}] = \frac{\P_{\Omega}\left[\beta_1 \TT_1+\beta_2\TT_2+\beta_3\TT_3\right]}{\beta_1+4(\beta_2+\beta_3)},\\
    \P_{\Omega^\perp}[\L^{t+1}] = \frac{\P_{\Omega^\perp}\left[\beta_2\TT_2+\beta_3\TT_3\right]}{4(\beta_2+\beta_3},
\end{gather}
where $\TT_1=\Y-\S^t+\Lambda_1^t$, $\TT_2=\sum_{n=1}^4\Lx^{n,t}-\Lambda_2^{n,t}$ and $\TT_3=\sum_{n=1}^4\Laux^{n,t}+\Lambda_3^{n,t}$.

\textbf{2. $\Lx^n$ update:}
The variables $\Lx^n$ can be updated using:
\begin{gather}
    \Lx^{n,t+1} = \argmin_{\Lx^n} \psi_n\|\Lx_{(n)}^n\|_*+\frac{\beta_2}{2}\|\Lx^n-\L^{t+1}-\Lambda_{2}^{n,t}\|_F^2,
\end{gather}
which is solved by a soft thresholding operator on singular values of $\left(\L^{t+1}+\Lambda_2^{n,t}\right)_{(n)}$  with a threshold of $\psi_n/\beta_2$. 

\textbf{3. $\Laux^n$ update:} The variables $\Laux^n$ can be updated using:
\begin{gather}
    \Laux^{n,t+1}=\argmin_{\Laux^n} \theta \mathrm{tr}\left({\Laux^n_{(n)}}^\top \Phi^n\Laux^n_{(n)}\right) + \frac{\beta_3}{2}\|\L^{t+1}-\Laux^n-\Lambda_3^{n,t}\|_F^2,
\end{gather}
which is solved by:
\begin{gather}
    \Laux_{(n)}^{n,t+1} = \beta_3 G_{inv}\left(\L^{t+1}-\Lambda_3^{n,t}\right)_{(n)},
\end{gather}
where  $G_{inv}=\left(\theta \Phi^n+\beta_3 \I\right)^{-1}$ always exists and can be computed outside the loop for faster update.

\textbf{4. $\S$ update:}
The variable $\S$ can be updated using:
\begin{gather}
    \S^{t+1} = \argmin_{\S} \lambda\|\S\|_1+\frac{\beta_1}{2}\|\P_{\Omega}[\S+\L^{t+1}-\Y-\Lambda_1^t]\|_F^2+\frac{\beta_5}{2}\|\S-\W^t-\Lambda_5^t\|_F^2,
\end{gather}
where the Frobenius norm terms can be combined and the expression can be simplified into:
\begin{gather}
    \P_\Omega[\S^{t+1}] = \argmin_{\P_\Omega[\S]} \|\P_\Omega[\S]\|_1 +\frac{\beta_1+\beta_5}{2\lambda}\left\|\P_\Omega[\S-\TT_s]\right\|_F^2,\\
    \P_{\Omega^\perp}[\S^{t+1}] = \argmin_{\P_{\Omega^\perp}[\S]} \|\P_{\Omega^\perp}[\S]\|_1 +\frac{\beta_5}{2\lambda}\left\|\P_{\Omega^\perp}[\S-\TT_s]\right\|_F^2,
\end{gather}
where 
\begin{gather}
    \P_{\Omega}[\TT_s]=\P_{\Omega}\left[\frac{\beta_1(\Y-\L^{t+1}+\Lambda_1^t)+\beta_5(\W^t+\Lambda_5^t)}{\beta_1+\beta_5}\right]\nonumber\\
    \P_{\Omega^\perp}[\TT_s] = \P_{\Omega^\perp}[\W^t+\Lambda_5^t]\nonumber
\end{gather} 
The above is solved by setting $\P_{\Omega}[\S^{t+1}]=\eta(\P_{\Omega}[\TT_s], \frac{\lambda}{\beta_1+\beta_5})$ and $\P_{\Omega^\perp}[\S^{t+1}]=\eta(\P_{\Omega^\perp}[\TT_s], \frac{\lambda}{\beta_5})$,
where $\eta(\a, \phi) = sign(\a)\odot max(|\a|-\phi, 0)$ and $\odot$ is elementwise or Hadamard product.

\textbf{5. $\W$ update:} The auxiliary variable $\W$ can be updated using:
\begin{gather}
    \W^{t+1}=\argmin_{\W}\frac{\beta_4}{2}\|\W\times_1 \Delta-\Z^t-\Lambda_4^t\|_F^2+\frac{\beta_5}{2}\|\S^{t+1}-\W-\Lambda_5^t\|_F^2,
\end{gather}
which is solved analytically by taking the derivative of the expression given above and setting it to zero which results in:
\begin{gather}
    \W^{t+1}_{(1)} = W_{inv}\left(\beta_5(\S-\Lambda_5)_{(1)}+\beta_4 \Delta^\top(\Lambda_4+\Z)_{(1)}\right),
\end{gather}
where $W_{inv}=\left(\beta_5\I+\beta_4 \Delta^\top \Delta\right)^{-1}$ always exists and can be computed outside the loop for faster update.

\textbf{6. $\Z$ update:} The auxiliary variable $\Z$, can be updated using:
\begin{gather}
    \Z^{t+1} = \argmin_{\Z} \gamma\|\Z\|_{1}+\frac{\beta_4}{2}\|\W^{t+1}\times_1 \Delta-\Z-\Lambda_4^t\|_F^2
\end{gather}
which is solved by $\eta(\W^{t+1}\times_1 \Delta-\Lambda_4^t,\gamma/\beta_4)$.

\textbf{7. Dual updates:} Finally, dual variables $\Lambda_1, \Lambda_2^n, \Lambda_3^n, \Lambda_4, \Lambda_5$ are updated using:
\begin{gather}
    \Lambda_1^{t+1} = \Lambda_1^t-\P_\Omega[\L^{t+1}+\S^{t+1}-\Y],\\
    \Lambda_2^{n,t+1} = \Lambda_2^{n,t}-(\Lx^{n,t+1}-\L^{t+1}),\\
    \Lambda_3^{n,t+1} = \Lambda_3^{n,t}-(\L^{t+1}-\Laux^{n,t+1}),\\
    \Lambda_4^{t+1} = \Lambda_4^t-(\W^{t+1}\times_1 \Delta-\Z^{t+1}),\\
    \Lambda_5^{t+1} = \Lambda_5^t-(\S^{t+1}-\W^{t+1}).
\end{gather}
The pseudocode for the proposed algorithm, GLOSS, is given in Algorithm \ref{alg:gloss}. The optimization for LOSS can be similarly computed without the updates on graph regularization and related variables $\{\Laux\}, \{\Lambda_3\}$.
\renewcommand{\algorithmicrequire}{\textbf{Input:}}
\renewcommand{\algorithmicensure}{\textbf{Output:}}
\begin{algorithm}  
\caption{GLOSS}
\begin{algorithmic}
\REQUIRE $\Y \in \mathbb{R}^{I_1\times I_2\times I_3\times I_4}$, $\Omega$, $\Phi$, parameters $\lambda$, $\gamma$, $\theta$, $\{\psi\}$, $\beta_1$, $\beta_2$, $\beta_3$, $\beta_4$, $\beta_5$, $\text{T}$.
\ENSURE $\L$ : Low-rank tensor;  $\S$: Sparse tensor.
\STATE Initialize $\S^0=0$, $\Z^0=0$, $\Lx^{n,0}=0$, $\Laux^{n,0}=0$, $\Lambda_1^0=0$, $\Lambda_2^{n,0}=0$, $\Lambda_3^{n,0}=0$, $\Lambda_4^0=0$, $\Lambda_5^0=0$, $\forall i \in \{1,\dots,4\}$.
\STATE{$W_{inv} \gets \left(\beta_5\I+\beta_4 \Delta^\top \Delta\right)^{-1}$}
\STATE{$G_{inv} \gets \left(\theta \Phi^n+\beta_3 \I\right)^{-1}$}
\FOR{$t=0$ \TO $\text{T}$}
\STATE{$\TT_1 \gets \Y-\S^t+\Lambda_1^t$}
\STATE{$\TT_2 \gets \sum_{n=1}^4\Lx^{n,t}-\Lambda_2^{n,t}$}
\STATE{$\TT_3\gets\sum_{n=1}^4\Laux^{n,t}+\Lambda_3^{n,t}$}
\STATE{$\P_{\Omega}[\L^{t+1}] \gets \P_{\Omega}\left[\beta_1 \TT_1+\beta_2\TT_2+\beta_3\TT_3\right]/(\beta_1+4(\beta_2+\beta_3))$}
\STATE{$\P_{\Omega^\perp}[\L^{t+1}] \gets \P_{\Omega^\perp}\left[\beta_2\TT_2+\beta_3\TT_3\right]/4(\beta_2+\beta_3)$}
\FOR{$n=1$ \TO $4$}
\STATE{$[U, \Sigma, V] \gets$  SVD $\left(\L^{t+1}+\Lambda_2^{n,t}\right)_{(n)}$}
\STATE{$\hat{\sigma}_{i,i} \gets \max{(\sigma_{i,i}-\frac{\psi_n}{\beta_2},0)}\quad \forall i\in \{1,\dots,I_n\}$}
\STATE{$\Lx_{(n)}^{n,t+1} \gets U\hat{\Sigma}V^\top$}
\STATE{$\Laux_{(n)}^{n,t+1} \gets \beta_3 G_{inv}\left(\L^{t+1}-\Lambda_3^{n,t}\right)_{(n)} $}
\ENDFOR
\STATE{$\P_{\Omega}[\TT_s] \gets \P_{\Omega}\left[\frac{\beta_1(\Y-\L^{t+1}+\Lambda_1^t)+\beta_5(\W^t+\Lambda_5^t)}{\beta_1+\beta_5}\right]$}
\STATE{$\P_{\Omega^\perp}[\TT_s] \gets \P_{\Omega^\perp}[\W^t+\Lambda_5^t]$}
\STATE{$\P_{\Omega}[\S^{t+1}]\gets \eta(\P_{\Omega}[\TT_s], \frac{\lambda}{\beta_1+\beta_5})$}
\STATE{$\P_{\Omega^\perp}[\S^{t+1}]\gets \eta(\P_{\Omega^\perp}[\TT_s], \frac{\lambda}{\beta_5})$}
\STATE{$\W^{t+1}_{(1)}\! \gets\! W_{inv}\!\left(\beta_5(\S^{t+1}\!-\!\Lambda_5^t)_{(1)}\!+\beta_4 \Delta^\top(\Lambda_4^t+\Z^t)_{(1)}\right)$}
\STATE{$\Z^{t+1} \gets \eta(\W^{t+1}\times_1 \Delta-\Lambda_4^{t},\frac{\gamma}{\beta_4})$}
\STATE{$ \Lambda_1^{t+1} \gets \Lambda_1^t-\P_\Omega[\L^{t+1}+\S^{t+1}-\Y] $}
\STATE{$\Lambda_4^{t+1} \gets \Lambda_4^t-(\W^{t+1}\times_1 \Delta-\Z^{t+1})$}
\STATE{$\Lambda_5^{t+1} \gets \Lambda_5^t-(\S^{t+1}-\W^{t+1})$}
\FOR{$n=1$ \TO $4$}
\STATE{$\Lambda_2^{n,t+1} \gets \Lambda_2^{n,t}-(\Lx^{n,t+1}-\L^{t+1}) $}
\STATE{$\Lambda_3^{n,t+1} \gets \Lambda_3^{n,t}-({\L}^{t+1}-\Laux^{n,t+1}) $}
\ENDFOR
\ENDFOR 
\end{algorithmic}
\label{alg:gloss}
\end{algorithm}
\subsection{Convergence}
\label{sec:conv}
In this section, we analyze the convergence of the proposed algorithm. First, we show that the proposed optimization problem can be written as a two-block ADMM. In previous work, linear and global convergence of ADMM is proven for two-block systems \cite{deng2016global} with no dependence on the hyperparameters $\beta_i$. We use this proof to derive a sketch of the proof of convergence for GLOSS. In the following discussion, we will assume that there is no missing data, {i.e.} $\P_{\Omega}[\Y]=\Y$, to simplify the notation. 
\begin{proposition}
Let $h(\S) = \lambda\|\S\|_1$, $j(\Z)= \gamma\|\Z\|_1$, $f(\{\Lx\})= \sum_{n=1}^4\psi_n\|\Lx_{(n)}^n\|_*$, $g(\{\Laux\})= \sum_{n=1}^4\mathrm{tr} \left({\Laux^n_{(n)}}^\top \Phi^n \Laux^n_{(n)}\right)$,  (\ref{eq:opt_f}) can be rewritten as a two-block ADMM problem as:
\begin{gather}
    \min_{V_1, V_2} f_1(V_1)+f_2(V_2),\quad
    \text{s.t.} \; B_1 V_{1}+B_2 V_{2} = C,
\end{gather}
where $f_1(V_1) = 0$ and $f_2(V_2) = h(\S)+j(\Z)+f(\{\Lx\})+\theta g(\{\Laux\})$ are both convex, $V_1 = [\L_{(1)}, \W_{(1)}]$, $V_2 = [\X_{(1)}, \text{cat}_1(\{\Lx\}), \text{cat}_1(\{\Laux\}), \Z_{(1)}]$ and $C= \text{cat}_1(\{\Y,0,\dots,0\})$. 
\end{proposition}
\begin{proof}
Using the assumptions above, (\ref{eq:opt_f}) can be rewritten as:
\begin{gather}
    \min_{\{\Lx\},\{\Laux\},\S,\Z} h(\S)+j(\Z)+f(\{\Lx\})+\theta g(\{\Laux\}),\nonumber \\
    A_1\L_{(1)}+A_2 \S_{(1)} + A_3\text{cat}_1(\{\Lx\})+A_4 \text{cat}_1(\{\Laux\})+ \nonumber \\A_5\W_{(1)}+ A_6 \Z_{(1)}=\text{cat}_1(\{\Y,0,\dots,0\}),
\end{gather}
where
\setcounter{MaxMatrixCols}{11}
\begin{gather}
    A_1 = 
    \begin{bmatrix}
    \I & \I & \I & \I& \I & \I & \I & \I & \I & 0 & 0
    \end{bmatrix}^\top,
    \nonumber\\
    A_2 = 
    \begin{bmatrix}
    \I & 0 & 0 & 0 & 0 & 0 & 0 & 0 & 0 & \I & 0
    \end{bmatrix}^\top,\nonumber \\
    A_3 = 
    \begin{bmatrix}
    0 & -\I & 0  & 0  & 0  & 0 & 0 & 0 & 0 & 0 & 0 \\
    0 & 0  & -\I & 0  & 0  & 0 & 0 & 0 & 0 & 0 & 0 \\
    0 & 0  & 0  & -\I & 0  & 0 & 0 & 0 & 0 & 0 & 0 \\
    0 & 0  & 0  & 0  & -\I & 0 & 0 & 0 & 0 & 0 & 0
    \end{bmatrix}^\top,\nonumber 
\end{gather}
\vspace{-1em}
\begin{gather}
    A_4 = 
    \begin{bmatrix}
        0 & 0 & 0 & 0 & 0 & -\I & 0 & 0 & 0 & 0 & 0 \\
        0 & 0 & 0 & 0 & 0 & 0 & -\I & 0 & 0 & 0 & 0 \\
        0 & 0 & 0 & 0 & 0 & 0 & 0 & -\I & 0  & 0 & 0 \\
        0 & 0  & 0  & 0  & 0  & 0 & 0 & 0 & -\I & 0 & 0
    \end{bmatrix}^\top,\nonumber \\
    A_5 = 
    \begin{bmatrix}
    0 & 0 & 0 & 0 & 0 & 0 & 0 & 0 & 0 & -\I & \Delta
    \end{bmatrix}^\top, \nonumber\\
    A_6 = 
    \begin{bmatrix}
    0 & 0 & 0 & 0 & 0 & 0 & 0 & 0 & 0 & 0 & -\I
    \end{bmatrix}^\top.\nonumber 
\end{gather}
From this reformulation, it is easy to see that $A_1^\top A_5=0$ and $A_2, A_3, A_4, A_6$ are all orthogonal to each other, {i.e.} $A_i^\top A_j=0$ where $i, j \in \{2,3,4,6\}$ and $i\neq j$.

Let 
matrices $B_1 = [A_1, A_5]$  and $ B_2 = [A_2, A_3, A_4, A_6]$. Since the variable $V_1$ consists of $\L$ and $\W$, the order of updating $\S$ and $\W$ needs to change in the new formulation as $A_2$ and $A_5$ are not orthogonal. In other words, while the updates of  $\{\Lx\}$ and $\{\Laux\}$ have no effect on the update of $\W$, this is not true for $\S$. Therefore, updating $\W$ before $\S$ might affect the solution. However, it was proven in \cite{yan2016self} that the change in the order of updates gives the equivalent solution if either one of the functions of the variables $\S$ and $\W$ is affine. This is true for \eqref{eq:opt_f}, as the function corresponding to $\W$ is a constant. Thus, the problem reduces to the two-block form:
\begin{gather}
    \min_{V_1, V_2} f_1(V_1)+f_2(V_2),\quad
    \text{s.t.} \; B_1 V_{1}+B_2 V_{2} = C.
\end{gather}
\end{proof}

\begin{theorem}
Let $\psi_n, \theta, \lambda, \gamma \geq 0$. The sequence $\{\L^t, \S^t\}$ generated by Algorithm \ref{alg:gloss} converges to the optimal solution of \eqref{eq:opt_f} $\{\L^*, \S^*\}$. 
\end{theorem}
\begin{proof}
By Proposition 1, \eqref{eq:opt_f} can be written as two-block ADMM formulation. The proof of  convergence of two-block ADMM follows \cite{deng2016global}.
\end{proof}

\subsection{Computational Complexity}
Let $\Y$ be a mode $N$ tensor  with number of elements $I=\prod_{n=1}^NI_n$. The computational complexity of each iteration of ADMM is computed as follows:
\begin{enumerate}
    \item The update of $\L$ only involves element-wise operations, thus the computational complexity is linear, i.e. $\O(I)$.
    \item The computational complexity of updating each $\Lx^n$ is $\O\left(I_nI\right)$. Thus, the total computational complexity is $\O(\sum_{n=1}^NI_nI)$. However, it is possible to reduce this cost by parallelizing the updates across modes, hence the complexity becomes $\O\left(\max_n(I_n)I\right)$. 
    \item The update of each $\Laux^n$ requires the computation of a matrix inverse with complexity $O(I_n^3)$ and a matrix multiplication with complexity $\O(I)$. As mentioned before, the inverse always exists and can be computed outside the loop.  Similar to the updates of $\Lx^n$, the total complexity is $\O(NI)$. This can be reduced  to $\O(I)$ by parallelizing across modes.
    \item The second update requires a soft thresholding, which has linear complexity, {i.e.}, $\O(I)$.
    \item The computational complexity of updating $\W$ is governed by matrix multiplication resulting in $\O(I)$ complexity.
    \item The update of $\Z$ consists of a matrix product followed by soft thresholding, which results in a complexity of $\O(I)$.
    \item The updates of the dual variables  do not require any additional multiplication operations. Thus, the computational complexity is negligible.
\end{enumerate}
It can be concluded that the complexity of each loop is governed by the updates of $\Lx^n$. Assuming that the algorithm iterates  $\text{T}$ times, the total computational complexity of the algorithm is $\O(\text{T}\sum_{n=1}^NI_nI)$, therefore, close to linear complexity in the number of elements. 
Since this complexity is dominated by the nuclear norm minimization, LOSS, HoRPCA and WHoRPCA also have similar computational complexities. 

\subsection{Anomaly Scoring}

The method proposed in this paper focuses on extracting spatiotemporal features for anomaly detection.  After extracting the features, {i.e.} the sparse part,  a baseline anomaly detector can be applied to obtain an anomaly score. In this paper, we evaluated three anomaly detection methods: Elliptic Envelope (EE) \cite{rousseeuw1999fast}, Local Outlier Factor (LOF) \cite{breunig2000lof} and One Class SVM (OCSVM) \cite{scholkopf2000support}.  These three methods are used to assign an anomaly score to each element of the sparse tensor. Each method was applied to all third mode fibers which correspond to different weeks' traffic activity. This is equivalent to fitting a univariate distribution to each of the third mode fibers of the tensor. The anomaly scores were used to create an anomaly score tensor. 
Finally, the elements with the highest anomaly scores were selected as anomalous while the rest were determined to be normal. 

\section{Experiments}
\label{sec:Exp}

In this paper, we evaluated the proposed method on both real and synthetic datasets. We compared our method to regular HoRPCA and weighted HoRPCA (WHoRPCA) where the nuclear norm of each unfolding is weighted. In addition, to evaluate the effect of graph regularization term in \eqref{eq:finopt}, we compared with LOSS corresponding to the objective function in \eqref{eq:opt_lrt}. As our method is focused on feature extraction for anomaly detection, we also compared our method to baseline anomaly detection methods such as EE, LOF and OCSVM applied to the original tensor. After the feature extraction stage, unless noted otherwise, such as "GLOSS-LOF", EE was used as the default method for computing anomaly scores to compare the different approaches. The number of neighbors for LOF is selected as $10$ as this is the suggested lower-bound in \cite{zhang2016tensor}. The outlier fraction of OCSVM is set to $0.1$ as only the anomaly scores, not labels, generated by OCSVM are used in the experiments. The methods used for comparison and their properties are summarized in Table \ref{tab:methods}. 
\begin{table}[H]
\scriptsize
    \centering
    \begin{tabular}{l c c c c c}
    \toprule
         & LR & SP & WLR & TS & GR \\
         \midrule
         GLOSS & + & + & + & + & + \\
         LOSS &  + & + & + &+ &- \\
         WHoRPCA & + & + & + & - & - \\
         HoRPCA & + & + & - & - & - \\
         EE & N/A & N/A & N/A & N/A & N/A \\
         LOF & N/A & N/A & N/A & N/A & N/A \\
         OCSVM & N/A & N/A & N/A & N/A & N/A \\
     \bottomrule
    \end{tabular}
    \caption{Properties of anomaly detection methods used in the experiments. The acronyms refer to the different regularization terms in the cost function: (LR) low-rank, (SP) sparse, (WLR) weighted low-rank, (TS) temporal smoothness, (GR) graph regularization.}
    \label{tab:methods}
\end{table}

For each data set, a varying $K$ percent of the elements with highest anomaly scores are determined to be  anomalous. With varying $K$, ROC curves were generated for synthetic data and the mean area under the curve (AUC) was computed for $10$ random experiments. For real data, number of detected events were reported for varying $K$.  

\subsection{Data Description}
To evaluate the proposed framework, we use two publicly available datasets as well as synthetic data.
\subsubsection{Real Data}
The first dataset is NYC yellow taxi trip records\footnote{https://www1.nyc.gov/site/tlc/about/tlc-trip-record-data.page} for 2018. This dataset consists of  trip information such as the departure zone and time, arrival zone and time, number of passengers, tips for each yellow taxi trip in NYC. In the following experiments, we only use the arrival zone and time to collect the number of arrivals for each zone aggregated over one hour time intervals. We selected 81 central zones to avoid zones with very low traffic \cite{zhang2019decomposition}. Thus, we created a tensor $\Y$ of size $24\times 7 \times 52 \times 81$ where the first mode corresponds to hours within a day, the second mode corresponds to days of a week, the third mode corresponds to weeks of a year and the last mode corresponds to the zones. 

The second dataset is Citi Bike NYC bike trips\footnote{https://www.citibikenyc.com/system-data} for 2018. This dataset contains the departure and arrival station numbers, arrival and departure times and user id. In our experiments, we aggregated bike arrival data for taxi zones imported from the NYC yellow taxi trip records dataset, instead of using the original stations, to reduce the dimensionality and to avoid data sparsity. The resulting data tensor is of size $24\times 7 \times 52 \times 81$. Some statistical characteristics for both data sets are provided in Table \ref{tab:data_stat}.

\begin{table}[H]   
\scriptsize
    \centering
    \begin{tabular}{l c  c c}
    \toprule
         & & NYC Yellow Taxi & Citibike Bikeshare\\
         \midrule
        \multirow{4}{*}{Mean of Std. Dev.}& $n=1$& $157.8$ &$156.9$\\
        &$n=2$& $167.1$ & $392.8$\\
        &$n=3$&$168.5$&$64.0$\\
        &$n=4$&$112.0$&$49.0$\\
        \midrule
        \multicolumn{2}{c}{Sparsity Level} & $0.265$ & $0.165$\\
        \multicolumn{2}{c}{Maximum Value} & $1788$ & $661$ \\
        \multicolumn{2}{c}{Average Value} & $43.95$ & $23.16$ \\
        \bottomrule
    \end{tabular}
    \caption{Statistical characteristics of the real data sets.}
    \label{tab:data_stat}
\end{table}

\subsubsection{Synthetic Data Generation}

Evaluating urban events in a real-world setting is an open challenge, since it is difficult to obtain urban traffic data sets with ground truth information, {i.e.} anomalies are not known {\em a priori} and what constitutes an anomaly depends on the nature of the data and the specific problem. To be able to evaluate our method quantitatively, we  generate synthetic data and inject anomalies. Following \cite{zhang2019decomposition}, we generated a synthetic data set by taking the average of the NYC taxi trip tensor, $\Y$, along the third mode, {i.e.} across weeks of a year. We then repeat the resulting three-mode tensor such that for each zone, average data for a week is repeated 52 times. We multiply each element of the tensor by a Gaussian random variable with mean $1$ and variance $0.5$ to create variation across weeks. We then inject anomalies to the data in randomly selected time intervals during a day by adding or subtracting a constant $c$ times the average of that time interval. A higher $c$ value will generate more separable anomalies as the difference between the values of anomalous and normal intervals will be high. In our experiments, from $81$ zones, we select random time intervals of duration $7$ for $700$ days. Thus, in total, there are $4900$ anomalous time points for the whole synthetic data. 

\subsection{Parameter Setting}
In this section, we will discuss how the different parameters in \eqref{eq:finopt} are selected. Following \cite{goldfarb2014robust}, we set $\lambda=1/\sqrt{\max(I_1,\dots,I_N)}$ for HoRPCA and $\beta_1=\frac{1}{5 \text{std}(vec(\Y))}$, where $vec(\Y)\in\R^{I_1I_2...I_N}$ is the vectorization of $\Y$ and $\text{std(.)}$ is the standard deviation. The other $\beta$ parameters for all methods are set to be the same as $\beta_1$. The selection of $\beta$ parameters does not affect the algorithm performance but changes the convergence rate as mentioned in Section \ref{sec:conv}.

\begin{figure}[H]
    \centering
    \begin{subfigure}[b]{.32\textwidth}
        \centering
        \includegraphics[width=\columnwidth]{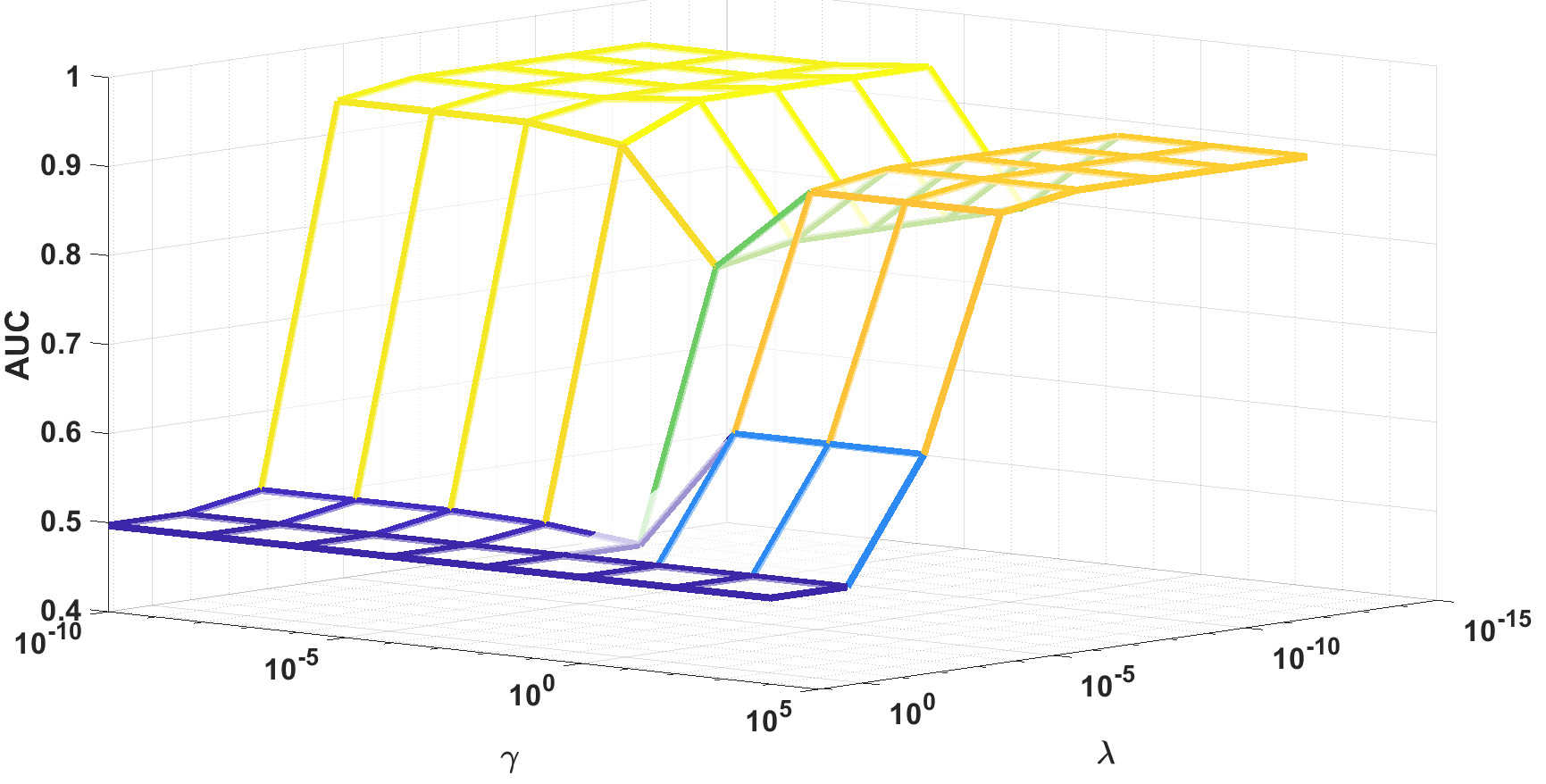}
        \caption{}
        \label{fig:param_gamvslam}    
    \end{subfigure}
    \begin{subfigure}[b]{.32\textwidth}
        \centering
        \includegraphics[width=\columnwidth]{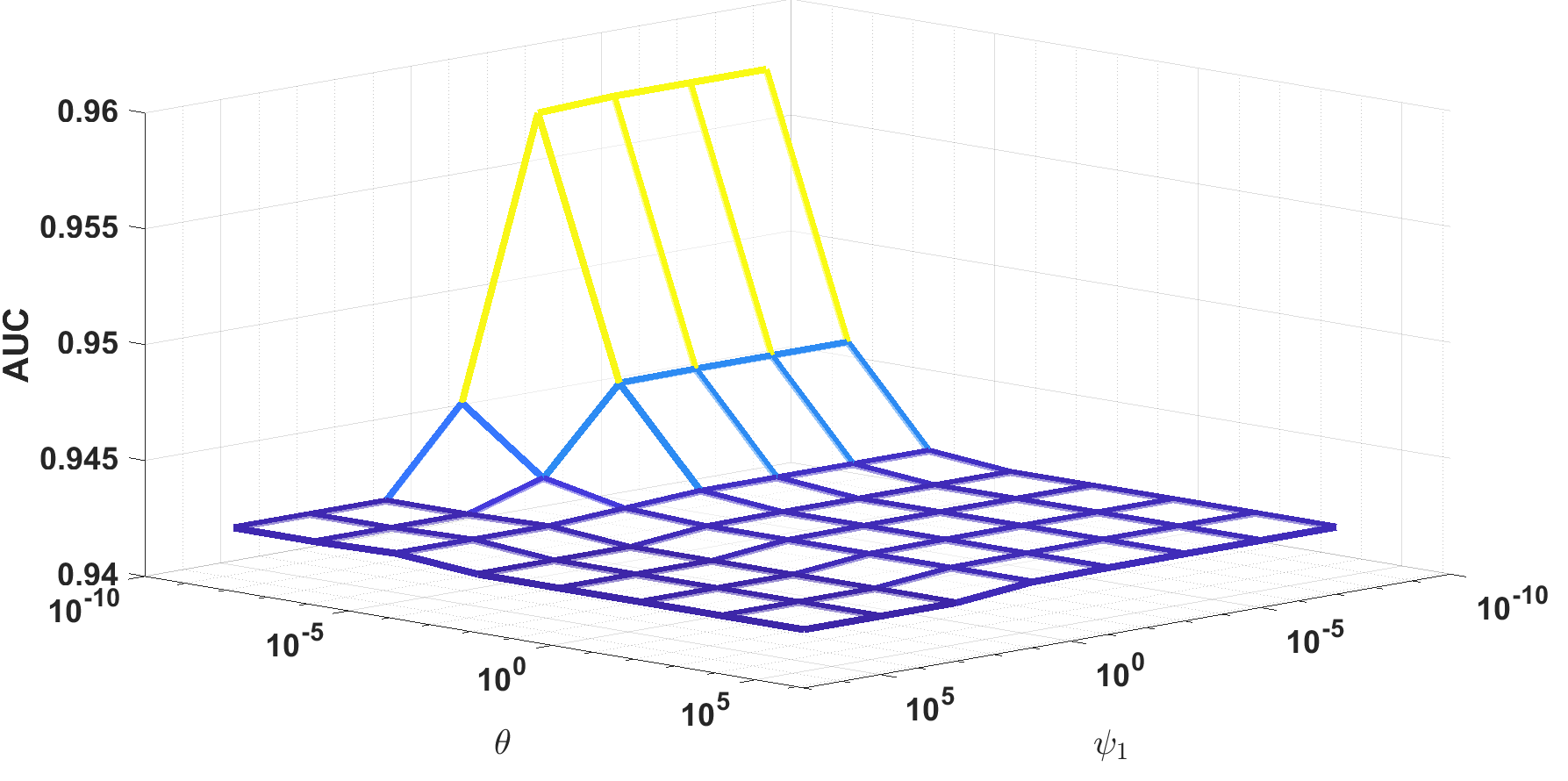}
        \caption{}
        \label{fig:param_thetavspsi1}    
    \end{subfigure}
    \begin{subfigure}[b]{.32\textwidth}
        \centering
        \includegraphics[width=\columnwidth]{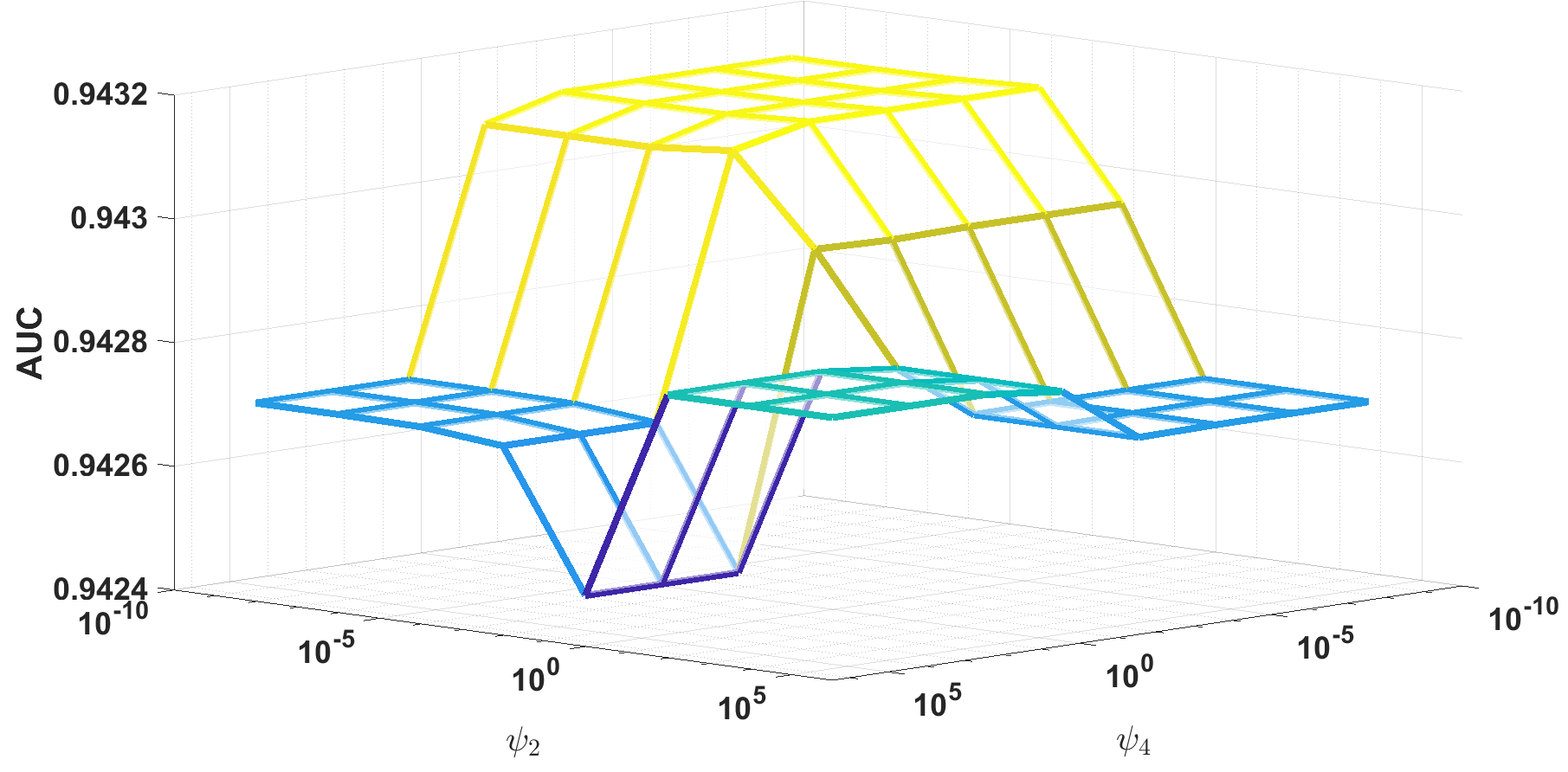}
        \caption{}
        \label{fig:param_psi2vs4}    
    \end{subfigure}
    \caption{Mean AUC values for various choices of: (a) $\lambda$ and $\gamma$, (b) $\theta$ and $\psi_1$, (c) $\psi_2$ and $\psi_4$. Mean AUC values across 10 random experiments are reported for each hyperparameter pair. For each set of experiments, the remaining hyperparameters are fixed.}
    \label{fig:param}
\end{figure}
Since the proposed algorithm requires the selection of multiple parameters, we perform a sensitivity analysis for the different hyperparameters. In Figure \ref{fig:param}, we present the average AUC for various ranges of the hyperparameters for GLOSS applied on the synthetic data generated with $c=2.5$. It can be seen from Figure \ref{fig:param_gamvslam} that while low values of $\lambda$ always provide the best results, $\gamma$ values are optimized  around $10^{-5}$. Increasing the sparsity penalty $\lambda$ above $10^{-3}$ results in a sparse tensor that is mostly  zero. At this $\lambda$ value, AUC becomes equal to 0.5 which is equivalent to randomly guessing the anomalous points. On the other hand, when $\gamma$ is too large, it smooths out all mode-1 fibers and generates the same anomaly score for each fiber. This is akin to identifying anomalous days, rather than time intervals. Figs. \ref{fig:param_thetavspsi1} and \ref{fig:param_psi2vs4} indicate that the performance does not vary much with the different choice of the weighted nuclear norm parameters and the choice of $\theta$ parameter. We repeated this analysis for different $c$ values and observed similar results indicating that the  proposed method is not sensitive to the selection of hyperparameters as long as they are selected within a certain range. Thus, we select $\lambda = \gamma = 1/\|\P_\Omega(\Y)\|_0$, where $\|\P_\Omega(\Y)\|_0$ is the number of nonzero elements of $\Y$ for GLOSS and $\lambda=\gamma=1/max(I_1,\dots, I_N)$ for LOSS and WHoRPCA. Since the ranks across each mode are closely related to the variance of the data within that mode, weights for each mode in the definition of nuclear norm, $\psi_n$s,  are selected to be inversely  proportional to the trace of the square root of the covariance matrix of mode-$n$, {i.e.} $\psi_{n}=\frac{p}{Tr(\sqrt{\Sigma_{Y_{(n)}})}}$, where $p$ is selected such that $\min_n(\psi_n)=1$ (For synthetic data with $c=2.5$, $n=3$). The parameter $\theta$ is set to be the geometric mean of $\psi_n$s, i.e. $\theta=\prod_{n=1}^4\psi_n^{1/4}$.
An important observation about the selection of the parameters is that depending on the data, the optimal values of hyperparameters might be different. This is due to the properties of the data such as size, variance and sparsity level. A cross-validation for determining the optimal parameters is not feasible as there are 7 parameters and ground truth labels for anomalies are rarely available for  real data. The parameters proposed in the previous paragraph provide good results for various tested data. 


\subsection{Experiments on Synthetic Data}
\subsubsection{Robustness against noise}
\begin{figure}[H]
    \centering
    \begin{subfigure}[b]{.32\textwidth}
        \includegraphics[width=.98\columnwidth]{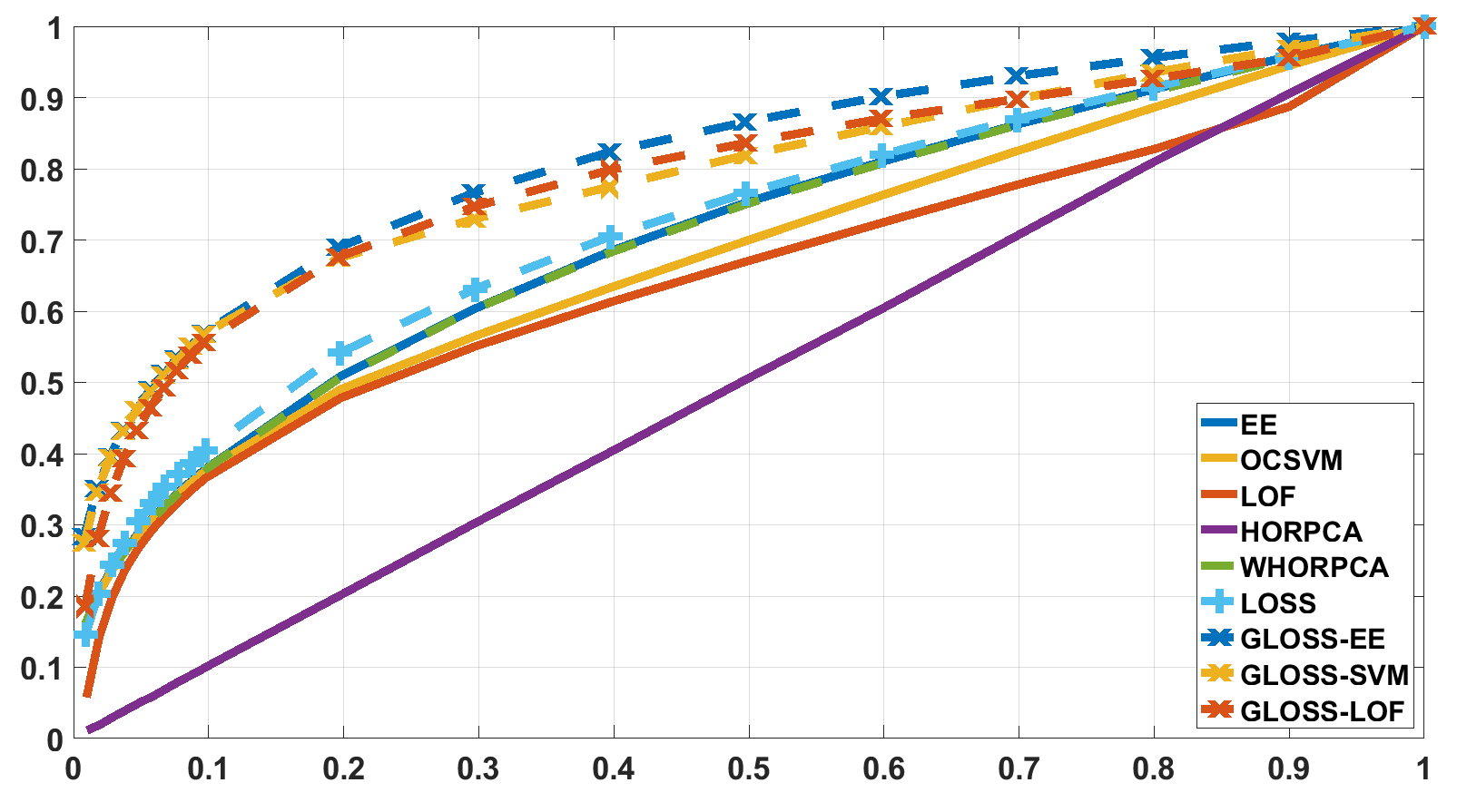}
        \caption{}
        \label{fig:amp1}
    \end{subfigure}
    \begin{subfigure}[b]{.32\textwidth}
        \includegraphics[width=.98\columnwidth]{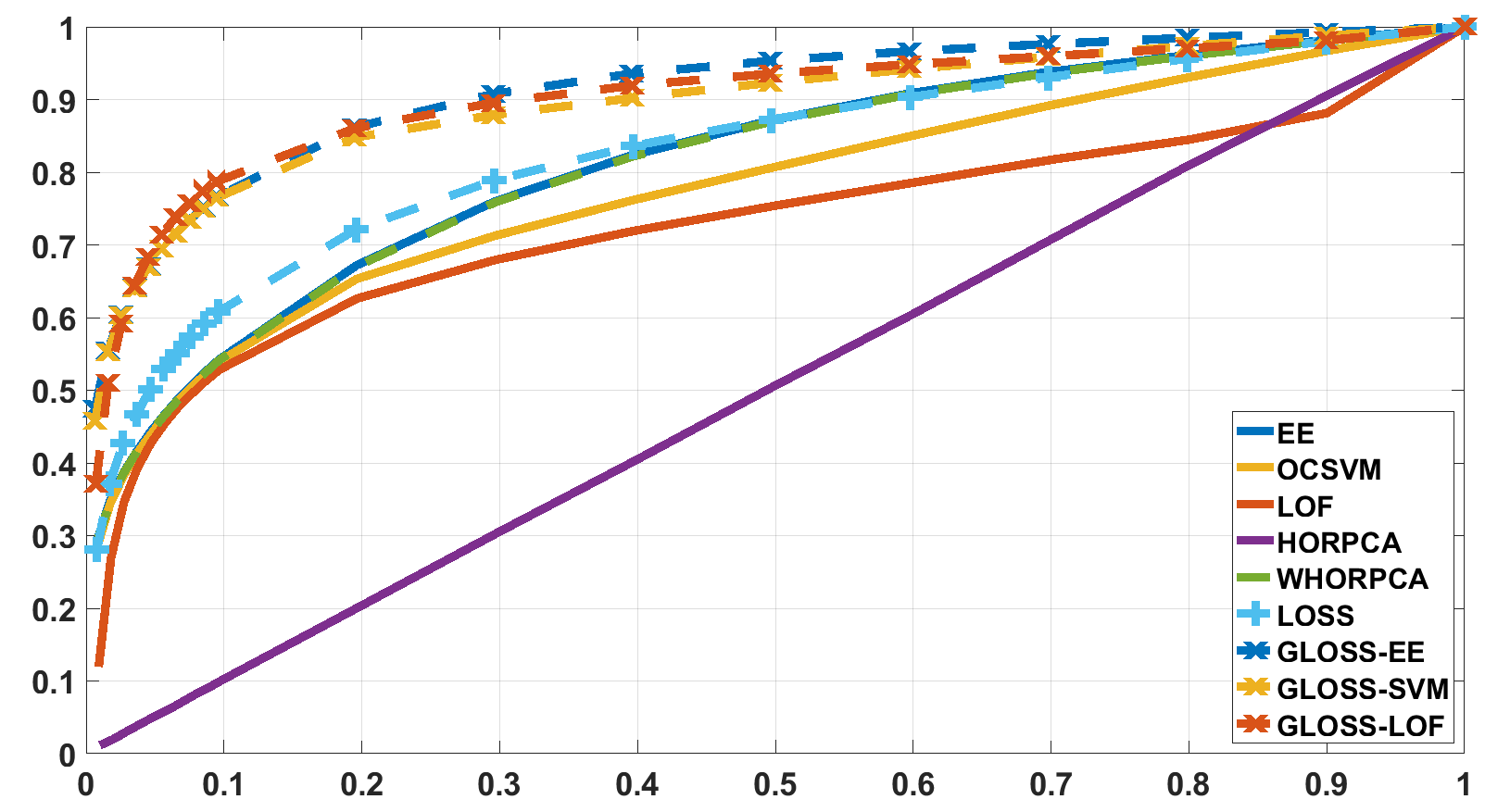}
        \caption{}
        \label{fig:amp2}
    \end{subfigure}
    \begin{subfigure}[b]{.32\textwidth}
        \includegraphics[width=.98\columnwidth]{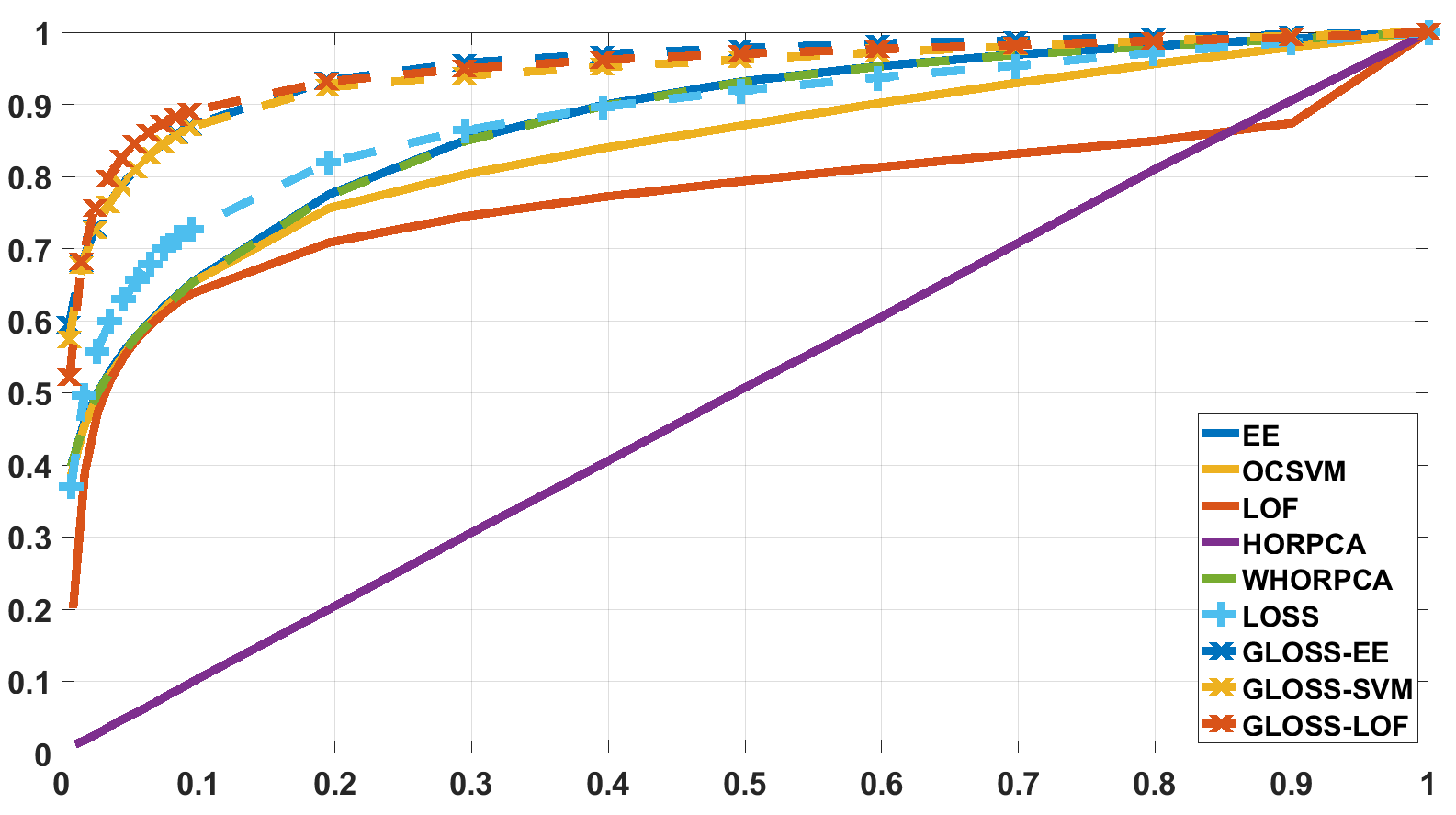}
        \caption{}
        \label{fig:amp3}
    \end{subfigure}
    \caption{ROC curves for various amplitudes of anomalies. Higher amplitude means more separability. $c=$ (a) $1.5$, (b) $2$, (c) $2.5$. ($P=0\%$)}
    \label{fig:amp}
\end{figure}
We first evaluated the effect of $c$, {i.e.} the strength of the anomaly, on the accuracy of the proposed method.  Low $c$ values imply that the amplitudes of anomalies are low and they may be indistinguishable from noise. From Figure \ref{fig:amp} and Table \ref{tab:nyc_auc}, it can be seen that for varying $c$ values, our method (GLOSS-EE) achieves the highest AUC values compared to both baseline methods and HoRPCA, WHoRPCA and LOSS. Among the remaining methods, LOSS performs better than WHoRPCA, especially when the anomaly strength is small. It is also important to note that the choice of the anomaly scoring method does not change the performance of GLOSS significantly. 

\begin{table}[H]
\scriptsize
    \begin{adjustbox}{width=\textwidth}
    \begin{tabular}{lccc|ccc}
        \toprule
        & $c=1.5$ & $c=2$ & $c=2.5$ & $P = 20\%$ & $P = 40\%$ & $P = 60\%$ \\
        \midrule
        EE      & $0.70\!\pm\! 0.003$ & $0.81\!\pm\!0.004$  & $0.87\!\pm\!0.002$ & $0.81\!\pm\!0.004$ & $0.61\!\pm\!0.008$ & $0.53\!\pm\!0.01$  \\
        LOF     & $0.67\!\pm\! 0.004$ & $0.78\!\pm\!0.004$  & $0.84\!\pm\!0.004$ & $0.79\!\pm\!0.005$ & $0.77\!\pm\!0.005$ & $0.73\!\pm\!0.007$ \\
        OCSVM   & $0.65\!\pm\! 0.005$ & $0.73\!\pm\! 0.003$ & $0.77\!\pm\!0.005$ & $0.81\!\pm\!0.004$ & $0.81\!\pm\!0.005$ & $0.76\!\pm\!0.006$ \\
        HoRPCA  & $0.5\!\pm\! 0.01$   & $0.5\!\pm\! 0.01$   & $0.5\!\pm\!0.01$   & $0.5\!\pm\!0.01$   & $0.5\!\pm\!0.01$   & $0.5\!\pm\!0.01$   \\
        WHoRPCA & $0.7\!\pm\! 0.004$  & $0.81\!\pm\! 0.004$ & $0.87\!\pm\!0.003$ & $0.81\!\pm\!0.004$ & $0.62\!\pm\!0.008$ & $0.56\!\pm\!0.01$  \\
        LOSS    & $0.72\!\pm\!0.01$   & $0.83\!\pm\!0.007$  & $0.88\!\pm\!0.007$ & $0.9\!\pm\!0.007$  & $0.78\!\pm\!0.01$  & $0.51\!\pm\!0.017$ \\
        GLOSS-EE & $\mathbf{0.82\!\pm\!0.005}$ & $\mathbf{0.91\!\pm\!0.003}$ & $\mathbf{0.95\!\pm\!0.002}$ & $0.93\!\pm\!0.002$ & $0.78\!\pm\!0.007$ & $0.65\!\pm\!0.008$ \\
        GLOSS-SVM & $0.79\!\pm\!0.006$ & $0.90\!\pm\!0.004$ & $0.94\!\pm\!0.002$ & $\mathbf{0.94\!\pm\!0.002}$ & $\mathbf{0.93\!\pm\!0.002}$ & $\mathbf{0.87\!\pm\!0.006}$ \\
        GLOSS-LOF & $0.79\!\pm\!0.006$ & $0.90\!\pm\!0.003$ & $\mathbf{0.95\!\pm\!0.002}$ & $0.88\!\pm\!0.003$ & $0.86\!\pm\!0.003$ & $0.85\!\pm\!0.004$\\
        \bottomrule
    \end{tabular}
    \end{adjustbox}
    \caption{Mean and standard deviation of AUC values corresponding to Figure \ref{fig:amp} (left) and Figure \ref{fig:pctg} (right). The proposed method, specifically GLOSS-SVM, outperforms the other algorithms in all cases significantly with  $p<0.001$. For varying $c$ values, $P$ is fixed at 0\%. Similarly, $c=2.5$ for varying $P$.}
    \label{tab:nyc_auc}
\end{table}
\subsubsection{Robustness against missing data}
In addition to injecting synthetic anomalies, we also remove varying number of days at random from the tensor to evaluate the robustness of the proposed method to missing data. After generating the synthetic data with $c=2.5$, a percentage $P$ of the mode-1 fibers is set to zero to simulate missing data, where the number of mode-1 fibers is equal to the total number of days, {i.e.} $7\times 52\times81$. The accuracy of anomaly detection for varying levels of missing data is illustrated in Fig. \ref{fig:pctg} and the corresponding AUC values (mean $\pm$ std) are given in Table \ref{tab:nyc_auc}.
\begin{figure}[H]
    \centering
    \begin{subfigure}[b]{.32\textwidth}
        \includegraphics[width=.98\columnwidth]{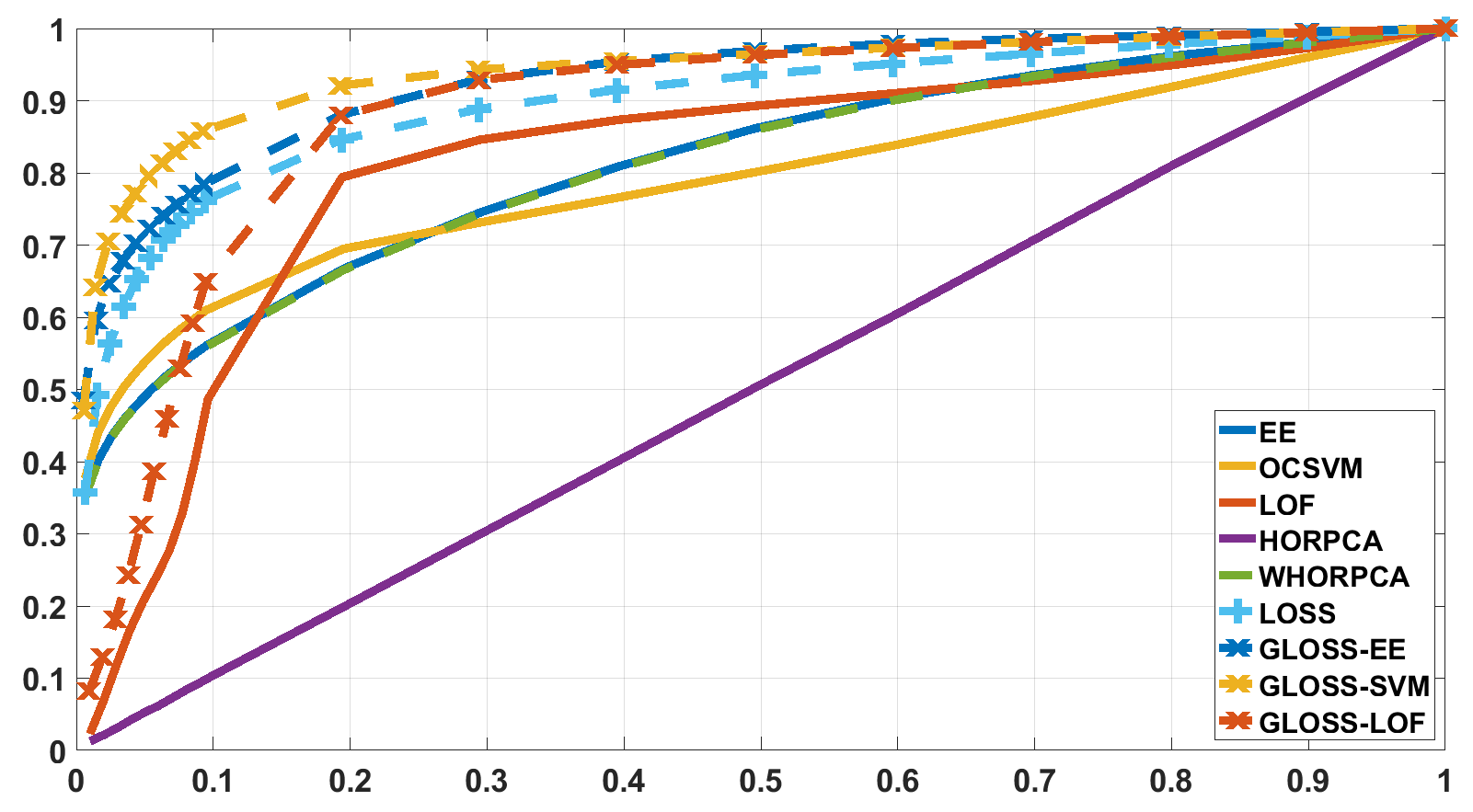}
        \caption{}
        \label{fig:miss20}
    \end{subfigure}
    \begin{subfigure}[b]{.32\textwidth}
        \includegraphics[width=.98\columnwidth]{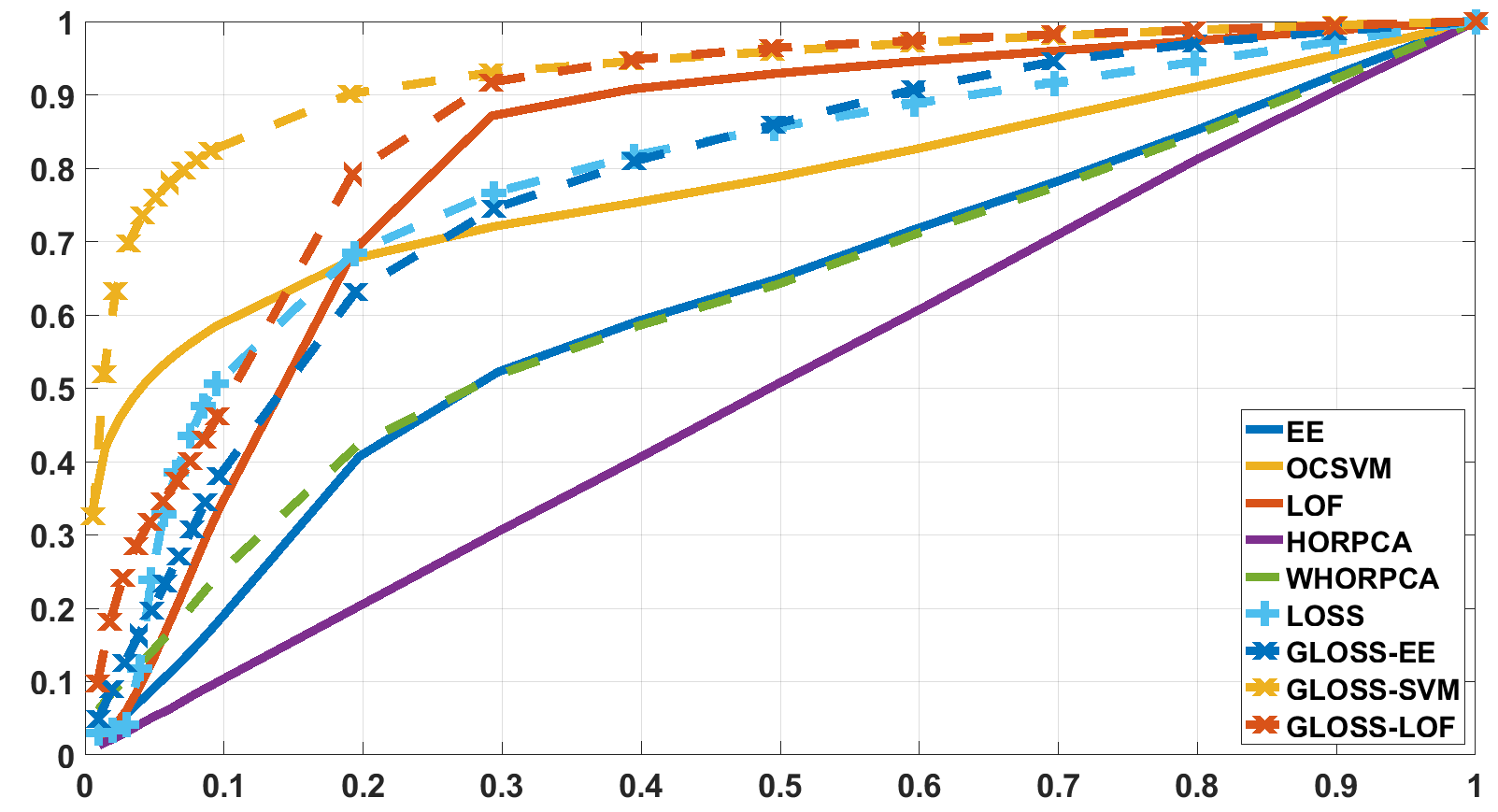}
        \caption{}
        \label{fig:miss40}
    \end{subfigure}
    \begin{subfigure}[b]{.32\textwidth}
        \includegraphics[width=.98\columnwidth]{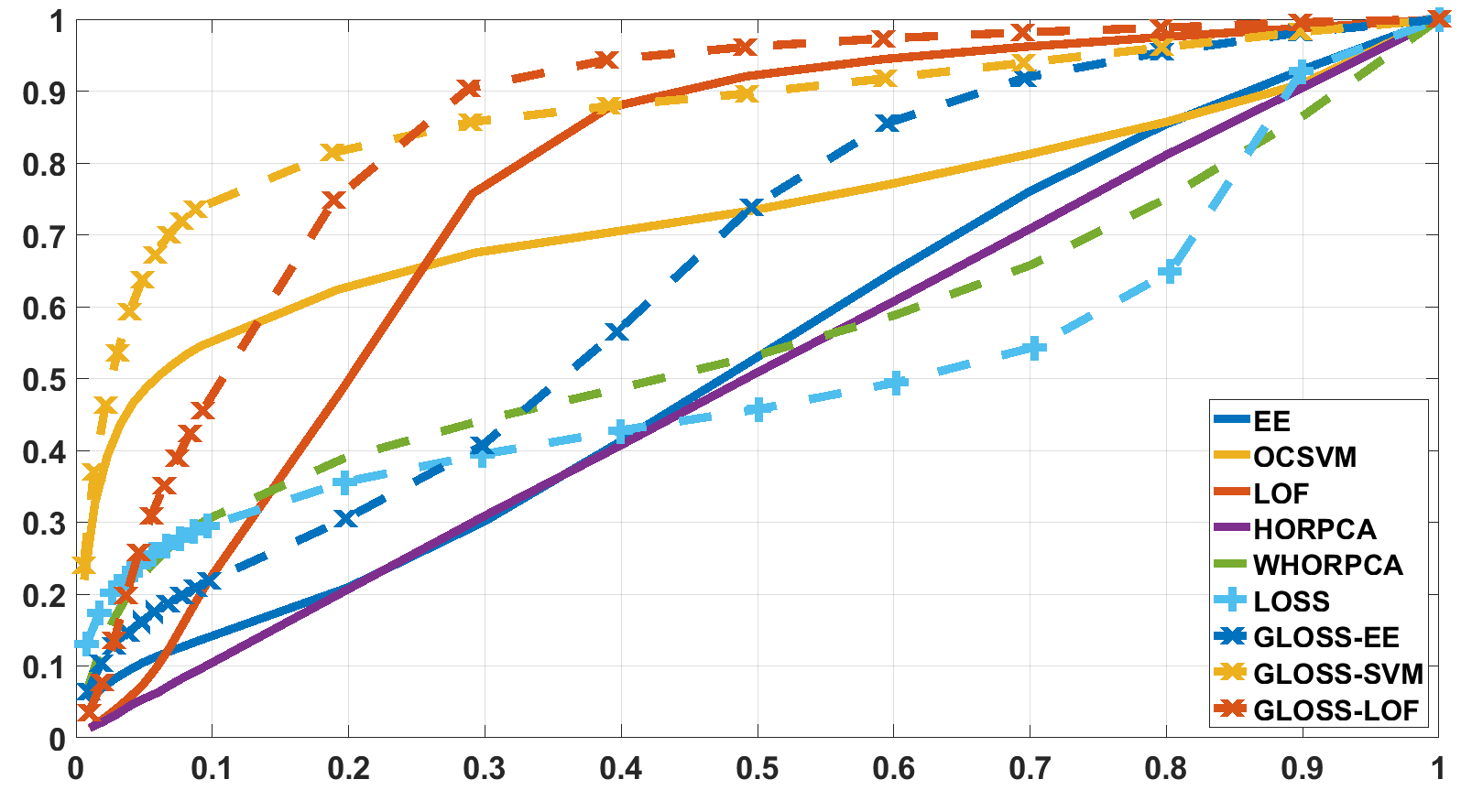}
        \caption{}
        \label{fig:miss60}
    \end{subfigure}
    \caption{ROC curves for varying percentage of missing data,  (a) $P=20\%$, (b) $P=40\%$, (c) $P=60\%$. ($c=2.5$)}
    \label{fig:pctg}
\end{figure}

While the performance of all the methods degrades with increasing levels of missing data, GLOSS provides the best anomaly detection performance and is robust against missing data compared to the other methods. In particular, GLOSS-SVM provides the highest accuracy. The proposed method provides better features for anomaly detection compared to using the original tensor or using the features from HoRPCA, WHoRPCA or LOSS. It can be seen that weighted low-rank approach improves the results compared to HoRPCA which provides the lowest accuracy. This is due to the fact that the data is not necessarily low-rank along each mode, thus focusing on modes with low variance improves the accuracy. Moreover, the selected $\lambda$ value for HoRPCA sets all entries of the sparse tensor to zero which is not useful for anomaly detection and provides an AUC of $0.5$. Incorporating temporal smoothness for the anomalies in the objective function lowers the false detection rate by penalizing instantaneous changes in traffic that do not constitute an actual anomaly. It is also interesting to note that while LOSS performs comparable to GLOSS for varying anomaly strength, the performance of LOSS degrades quickly with increasing missing data. Therefore, even though both WHoRPCA and LOSS are equipped to handle missing data, GLOSS is more robust as it uses side information in the form of similarity graphs. 

\subsection{Experiments on Real Data}
To evaluate the performance of the proposed methods on real data, we compiled a list of 20 urban events that took place in the important urban activity centers such as city squares, parks, museums, stadiums and concert halls, during 2018. We used the same set of urban events for both taxi and bike data. To detect the activities, top-$K$ percent, with varying K, of the highest anomaly scores of the extracted sparse tensors are selected as anomalies and compared against the compiled list of urban events. 
Detection performance for all methods is given in Tables \ref{tab:nyc_real} and \ref{tab:nyc_bike} for the taxi and bike data, respectively.
\begin{table}[H]
\scriptsize
    \centering
    \begin{tabular}{l c c c c c c c c c}
    \toprule
         \% & 0.014 & 0.07 & 0.14 & 0.3 & 0.7 & 1 & 2 & 3\\
         \midrule
         EE  & 0 & 0 & 1 & 5 & 8 & 9 & 15 & 18\\ 
         LOF & 0 & 0 & 0 & 1 & 1 & 2 & 3 & 5\\
         OCSVM & 0 & 0 & 0 & 3 & 9 & 11 & 14 & 16\\
         HoRPCA & 0 & 0 & 0 & 0 & 0 & 0 & 0 & 1\\
         WHoRPCA & 3 & 4 & 4 & 4 & 6 & 7 & 9 & 9\\ 
         LOSS & 0 & 0 & 1 & 5 & 11 & 14 & 16 & 19\\ 
         GLOSS-EE & 1 & 6 & 8 & 12 & 16 & 18 & 19 & 20\\ 
         GLOSS-LOF & 1 & 6 & 8 & 14 & 17 & 18 & 19 & 20\\
         GLOSS-SVM & 0 & 2 & 3 & 7 & 13 & 14 & 17 & 19\\ 
    \bottomrule
    \end{tabular}
    \caption{Results for 2018 NYC Yellow Taxi Data. Columns indicate the percentage of selected points with top anomaly scores. The table entries correspond to the number of events detected at the corresponding percentage.}
    \label{tab:nyc_real}
\end{table}
From Table \ref{tab:nyc_real}, it can be seen that anomaly scoring methods applied to the spatiotemporal features extracted by GLOSS perform the best for the NYC Taxi data. The performance of GLOSS is followed by LOSS as temporal smoothness allows for detection of events at lower $K$ by removing anomalies resulting from noise. WHoRPCA performs well initially but as more points are considered it fails to detect the event of interest.  HoRPCA performs the worst in both data sets as the optimization is not tailored for anomaly detection. Among the baseline methods, EE performs the best while LOF performs the worst. However, when the features extracted from GLOSS are input to LOF and EE, their performances become very similar. This shows that GLOSS is effective at separating anomalous entries from noise and normal traffic activity and thus, improves the performance of both LOF and EE. It is important to note that most of the anomalies cannot be detected at low $K$ values because events such as New Year's Eve or July 4th celebrations change the activity pattern in the whole city and constitute the majority of the anomalies detected at low $K$ values for Taxi Data.

\begin{table}[H]
\scriptsize
    \centering
    \begin{tabular}{l c c c c c c c c c}
    \toprule
         \% & 0.3 & 1 & 2 & 3 & 4.2 & 7 & 9.7 & 12.5\\
         \midrule
         EE  & 0 & 0 & 2 & 2 & 2 & 2 & 2 & 3 \\
         LOF  & 1 & 1 & 1 & 1 & 2 & 2 & 3 & 5\\
         OCSVM  & 0 & 0 & 1 & 2 & 2 & 2 & 4 & 5\\ 
         HoRPCA  & 0 & 0 & 0 & 0 & 0 & 0 & 0 & 0\\ 
         WHoRPCA  & 2 & 4 & 6 & 6 & 6 & 6 & 6 & 6\\ 
         LOSS & 1 & 2 & 2 & 2 & 9 & 9 & 10 & 14\\
         GLOSS-EE  & 1 & 1 & 3 & 5 & 9 & 13 & 14 & 14\\
         GLOSS-LOF & 1 & 2 & 2 & 3 & 3 & 3 & 3 & 3\\ 
         GLOSS-SVM  & 1 & 2 & 5 & 5 & 7 & 9 & 10 & 13\\ 
    \bottomrule
    \end{tabular}
    \caption{Results on 2018 NYC Bike Trip Data. }
    \label{tab:nyc_bike}
\end{table}
The performance of all methods is significantly reduced in Bike Data as can be seen from Table \ref{tab:nyc_bike}. This is because Bike Data is very noisy with a large number of days, or points that would be considered anomalous. Also, some of the selected events do not produce significant changes in Bike Data such as New Year's Eve as usage of bikes at midnight is low even though it's a significant event for taxi traffic. Changes in the weather also affect the performance by increasing the variance of the data, especially across the third mode, which corresponds to the weeks of the year. In Figure \ref{fig:bike_comp}, we illustrate the bike data for July 4th at Hudson River banks, and the low-rank and sparse parts extracted by GLOSS. It can be seen that as the data varies across different weeks, the low-rank part can explain this variance well by fitting a pattern to days with varying amplitudes. Thus, the proposed method does not get affected by events such as the weather as it can capture both low and high traffic days in the low-rank part which can be seen in Fig \ref{fig:case_lr}. The deviations from the daily pattern, rather than the actual traffic volume, is captured by the sparse part, which is then input to the anomaly scoring algorithms. Thus, our method is able to extract the events at a fairly low $K$.  

\begin{figure}[H]
    \centering
    \begin{subfigure}[b]{.32\textwidth}
        \includegraphics[width=.98\columnwidth]{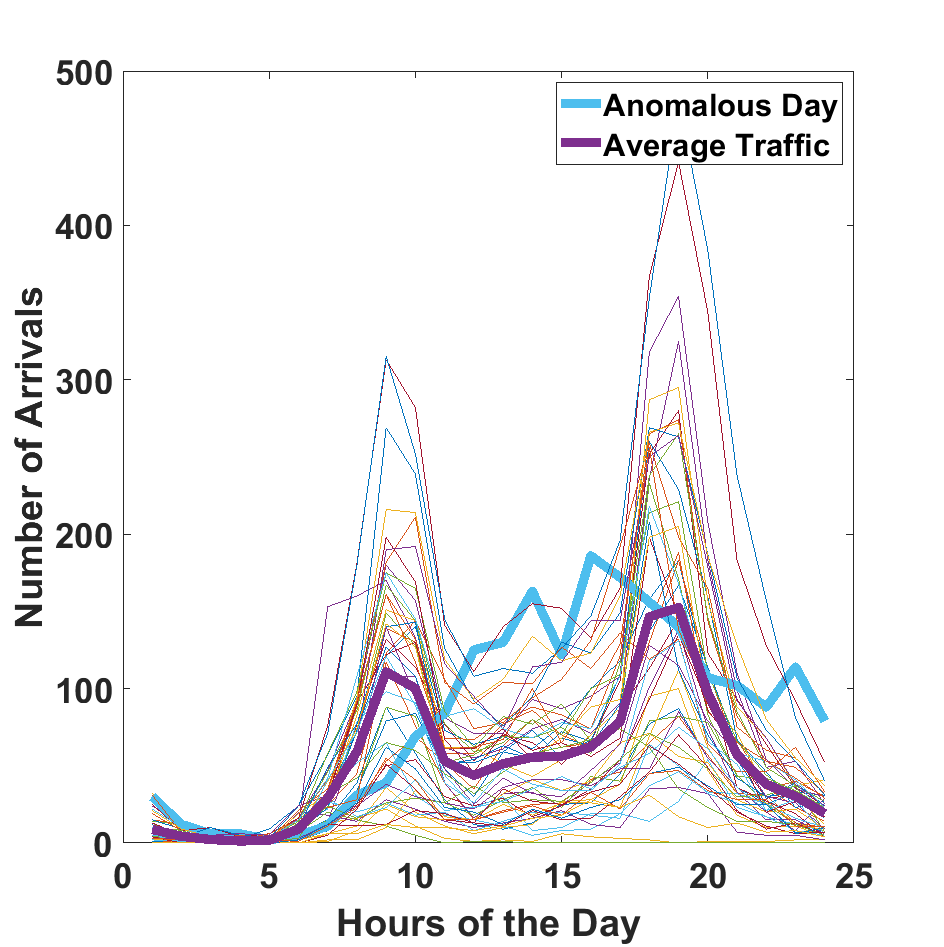}
        \caption{}
        \label{fig:case}
    \end{subfigure}
    \begin{subfigure}[b]{.32\textwidth}
        \includegraphics[width=.98\columnwidth]{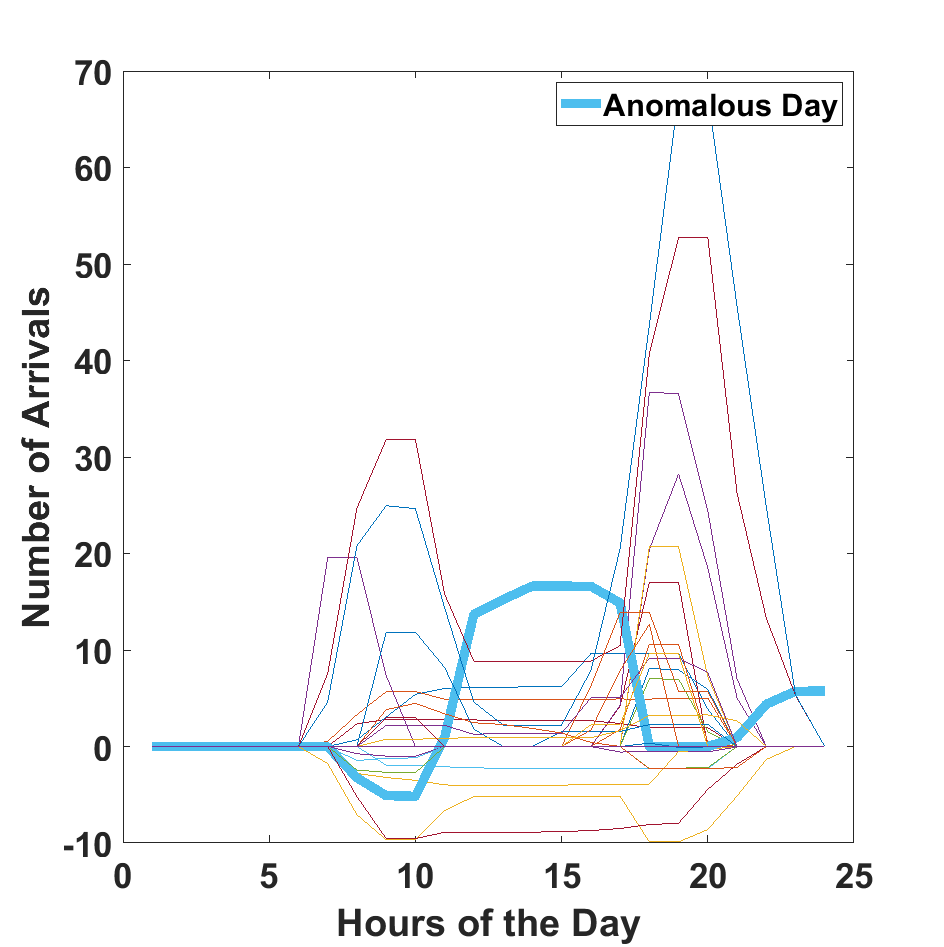}
        \caption{}
        \label{fig:case_spa}
    \end{subfigure}
    \begin{subfigure}[b]{.32\textwidth}
        \includegraphics[width=.98\columnwidth]{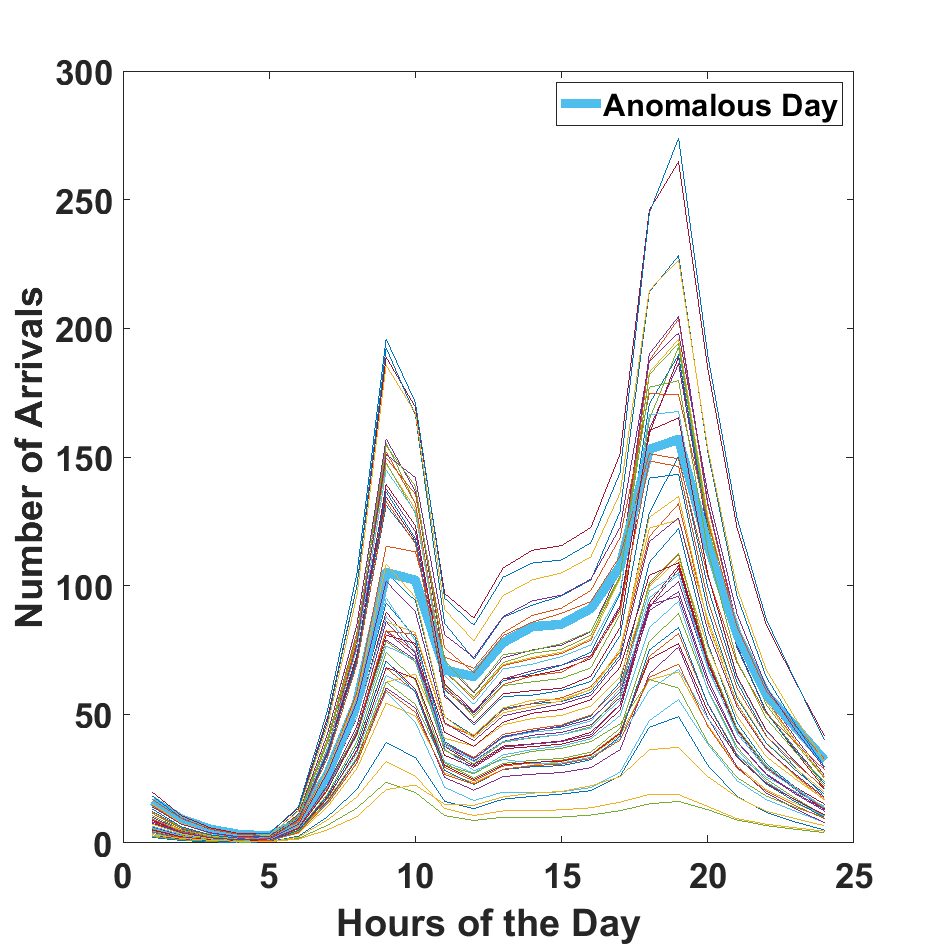}
        \caption{}
        \label{fig:case_lr}
    \end{subfigure}
    \caption{Bike Activity data, the extracted sparse part and low-rank part across for July 4th Celebrations at Hudson River banks. (a) Real Data where the traffic for 52 Wednesdays is shown along with the traffic on Independence Day and average traffic; (b) Sparse tensor where the curve corresponding to the anomaly is highlighted; (c) Low-rank tensor with the curve corresponding to the Independence Day highlighted.}
    \label{fig:bike_comp}
\end{figure}

\section{Conclusion}
In this paper, we proposed a robust tensor decomposition based anomaly detection method for urban traffic data. The proposed method extracts  a low-rank component using a weighted nuclear norm and imposes the sparse component to be temporally smooth to better model the anomaly structure. Finally, graph regularization is employed to preserve the geometry of the data and to account for nonlinearities in the anomaly structure. An ADMM  based computationally efficient and scalable algorithm is proposed to solve the resulting optimization problem. As the proposed method focuses on spatiotemporal feature extraction, the resulting features can be input to well-known anomaly detection methods such as EE, LOF and OCSVM for anomaly scoring.

The proposed method is evaluated on both synthetic and real urban traffic data. Results on synthetic data illustrate the robustness of our method to varying levels of missing data and its sensitivity to even low amplitude anomalies. In particular, our method outperforms WHoRPCA thanks to temporal smoothness assumption on the sparse part. Moreover, the graph regularization  improves the accuracy further by ensuring that the low-rank projections preserve local geometry of the data. In real data, our method begins to detect anomalies earlier, {i.e.} the top anomaly scores usually correspond to events of interest, than existing methods. GLOSS provides further improvement over LOSS as more events are detected for a given number of selected points. Furthermore, the results from real data show how the extracted sparse component highlights the anomalous activities. 

In future work, a statistical tensor anomaly scoring  method  will be explored instead of scoring each fiber individually by a separate algorithm. Applications and extensions of the proposed method on network data and other spatiotemporal data with different characteristics such as fMRI will also be considered.



%
\newpage
\bibliographystyle{IEEEtran}
\bibliography{main}

\end{document}